\documentclass{amsart}
\usepackage{amssymb}
\usepackage{amsfonts}
\usepackage{geometry}

\newtheorem{theorem}{Theorem}
\theoremstyle{plain}

\newtheorem{claim}{Claim}

\newtheorem{remark}{Remark}

\numberwithin{equation}{section}
\input{tcilatex}

\begin{document}
\date{August 17, 2022}
\title[Strategy-proof aggregation in median semilattices]{Strategy-proof
aggregation rules in median semilattices with applications to preference
aggregation}
\author{Ernesto Savaglio}
\address{DEA, University of Chieti-Pescara, Italy}
\author{Stefano Vannucci}
\address{DEPS, University of Siena, Italy}

\begin{abstract}
Two characterizations of the whole class of strategy-proof aggregation rules
on rich domains of locally unimodal preorders in finite median
join-semilattices are provided. In particular, it is shown that such a class
consists precisely of generalized weak sponsorship rules induced by certain
families of order filters of the coalition poset. It follows that the
co-majority rule and many other inclusive aggregation rules belong to that
class. The co-majority rule for an odd number of agents is characterized and
shown to be equivalent to a Condorcet-Kemeny median rule. Applications to
preference aggregation rules including Arrowian social welfare functions are
also considered. The existence of strategy-proof anonymous, weakly neutral
and unanimity-respecting social welfare functions which are defined on 
\textit{arbitrary} profiles of total preorders and satisfy a suitably
relaxed independence condition is shown to follow from our characterizations.

\textit{JEL\ Classification: }D71
\end{abstract}

\maketitle

\section{Introduction}

Aggregation rules are procedures to choose outcomes by taking
outcome-profiles as inputs. They are often encountered in the collective
decision literature under the labels of `voting schemes', `voting rules',
`social choice rules', `aggregators', `consensus functions', etc. The
outcomes to be aggregated usually are alternative items in some relevant
outcome space (e.g. scores, grades, signals, preferences, judgements, nodes
of an abstract network, etc.) which are submitted by the agents. Since
agents are stakeholders endowed with nonverifiable `preferential attitudes'
on the different items of the outcome space, then a \textit{reliable} and 
\emph{effective} decision protocol should be \emph{strategy-proof}, i.e.
immune to advantageous individual manipulations through submission of 
\textit{false} information. The celebrated Gibbard-Satterthwaite theorem
implies that if the domain of admissible individual preferences includes
every possible linear order on the outcome space and the range of a
strategy-proof aggregation rule consists of at least three alternatives,
then that aggregation rule must be a \emph{dictatorial rule}, certainly a
scarcely appealing protocol. Nevertheless, there exists a much wider and
interesting class of strategy-proof aggregation rules whenever \emph{rich},
i.e. suitably large, domains of certain \emph{single peaked} preference
preorders are considered. Single peaked preferences are preorders with a
unique maximum (the top element) which naturally arise if each agent's
representation of the outcome space is endowed with a plausible notion of
compromises (\textit{betweenness relation}), meaning that the outcome lying
between the top and another outcome cannot be worse then the latter. The
current literature has identified the outcome spaces (domains), on which
single peaked total preorders are defined, that allow strategy-proof rules.
We study here the case in which single peaked preorders are defined on
outcome spaces endowed with a (join-)\emph{semilatticial structure}.%
\footnote{%
We recall that a partial order over a set is a reflexive, antisymmetric and
transitive binary relation and that a partial order is a (join-)semilattice
if every pair of elements in the set has a least upper bound, named the join
or supremum of the pair, with respect to the underlying order relation.}

Thus, the present work is devoted to characterizing those \textit{%
aggregation rules in finite median\footnote{%
A join-semilattice is \textit{median }if the ternary partial operation $\mu $
such that $\mu (x,y,z)=(x\vee y)\wedge (y\vee z)\wedge (x\vee z)$ is a
well-defined operation (see Section 2 for more details).} join-semilattices}
which are \textit{strategy-proof }on some \textit{rich domains} of locally
unimodal (or \textit{single peaked}) total preorders on their outcomes.

We also characterize the \emph{co-majority} (or median) \emph{rule} as the
only one, within the class of such strategy-proof rules, that is \textit{%
anonymous} and \textit{bi-idempotent}\footnote{%
Namely, a rule selecting an outcome between the two proposals advanced by a
perfectly polarized population.} when the number of agents is odd.
Applications of our characterization results to \textit{preference}
aggregation rules are also provided, focussing on the classic case of
Arrowian \textit{social welfare functions}, namely aggregation rules taking
arbitrary profiles of total preference preorders as inputs and returning
total preference preorders as outputs.\footnote{%
It is worth emphasizing here that our usage of the term `Arrowian social
welfare function', while arguably sound and well-grounded, is by no means
widely established. Sometimes that term is also used to denote aggregation
rules for profiles of linear orders, possibly with the additional conditions
of Idempotence and the Arrowian `Independence of Irrelevant Alternatives'
requirement (see e.g. Sethuraman, Teo, Vohra (2003)).} Thus, the existence
of anonymous, unanimity-respecting and strategy-proof social welfare
functions\ on the full domain of total preorders is established. In
particular, a new characterization of \textit{generalized Condorcet-Kemeny
rules} for an odd number of agents is proved. Similar applications to
preference aggregation rules arising from other, less regular preference
structures, including generalized tournaments, reflexive relations, and
irreflexive relations are also briefly presented and discussed.

Now, addressing strategy-proofness issues for aggregation rules of any sort
requires a suitable specification of the agents' preferences on outcomes,
namely their \textit{preferences} \textit{on preferences. }It is well-known,
in view of the Gibbard-Satterthwaite `impossibility theorem', that (i) some 
\textit{domain-restriction on the foregoing `preferences on preferences'} is
required in order to open up the possibility to design interesting and
non-dictatorial strategy-proof preference-aggregation rules, and (ii)\ some
form of \textit{single-peakedness} is a most natural and plausible
domain-restriction to that effect. But single-peakedness notions typically
rely in turn on an underlying ternary \textit{betweenness relation }defined
on the `preference space' which is supposedly shared by the relevant agents
and thus should presumably be `naturally' embedded in that space. Therefore,
when it comes to the application of our basic model to preference
aggregation we are immediately confronted with a list of key issues to
address, namely:

\begin{itemize}
\item What sort of preference relations are to be aggregated? Arbitrary 
\textit{total preorders,} \textit{linear orders} or even wider domains
containing them? (Of course, the answer has a significant impact on the
general structure of the outcome space to focus on).

\item Are the preference profiles to be aggregated of an arbitrary but 
\textit{fixed finite size }(fixed population approach), or of \textit{every
possible finite size }(variable population approach)?

\item What \textit{type of aggregation protocol} are we to consider? Namely,
given a profile of preference relations of the prescribed type as input,
what kind of object is the output required to be? A single preference
relation of the prescribed type (\textit{aggregation} with no qualifier,
namely \textit{exact or pure aggregation}), one or more preference relations
of the prescribed type (\textit{multi-aggregation}), a single preference
relation belonging to a class which includes but does not reduce to the
prescribed type for the input \textit{((domain) restricted aggregation)}, or
a single preference relation of the same type as that prescribed for inputs
but enjoying some \textit{additional requirements ((codomain) constrained
aggregation)?}

\item What sort of \textit{single-peakedness} property for the relevant
`preferences on preferences' are we to focus on, or equivalently, what is
the most natural/plausible notion of \textit{betweenness} on the basic
`preference space' to refer to?
\end{itemize}

This paper relies on a definite choice of focus for each one of the
foregoing issues, namely:

(a) the basic preference domain should \textit{include all the total
preorders};

(b) the preference profiles to be aggregated are of some \textit{fixed size};

(c) the type of aggregation protocol to focus on is (pure) \textit{%
aggregation }(and possibly \textit{constrained aggregation)};

(d) the \textit{betweenness relation} to be used in order to define
single-peaked `preferences on preferences' should be \textit{the one
`naturally' dictated by the underlying basic preference domain }of the
aggregation rule under consideration.

Within such a framework, the main results established in the present work on
strategy-proof aggregation rules in median semilattices \textit{as applied
to preference aggregation }consists in:

\begin{description}
\item[$\left( i\right) $] (see Corollary 1 and Proposition 1)\textbf{\ }%
proving the existence of a large class of (`full-domain'!) social welfare
functions that are strategy-proof on the domain of all single-peaked
'preferences on preferences',\textit{\ }and characterizing them as those
aggregation rules $f$ for total preorders whose behaviour is dictated by a
family $F_{m}$ of superset-closed collections of agent-coalitions, one for
each bipartite total preorder $m$ of $A$ (namely, a preorder having just 
\textit{two indifference classes of good -respectively, bad- alternatives}),
in the following manner. At every preference profile of total preorders $%
R_{N}$, $f(R_{N})$ is just the intersection of all the bipartite total
preorders $m$ of $A$ that -at preference profile $R_{N}$- are consistent
with the preferences of \textit{all }the agents of \textit{at least one }%
coalition in $F_{m}$. Such a large class of strategy-proof social welfare
functions includes \textit{quorum systems}, \textit{fixed}-\textit{majority
collegial rules, weakly neutral rules, quota rules }and (within the latest
subclass) the \textit{co-majority (median) rule} (all of them can also be
described as \textit{weak sponsorship rules});

\item[$\left( ii\right) $] (see Proposition 2) proving that the co-majority
rule $\widehat{f}^{\partial maj}$is the only social welfare function that is 
\textit{strategy-proof} on the domain of preferences on the set of all total
preorders of a finite set of alternative social states that are
single-peaked, \textit{anonymous} (namely, invariant with respect to
agent-relabeling) and \textit{bi-idempotent }(namely, capable to choose one
of the pair of proposed preorders at any perfectly bipolarized profile).
Furthermore, under the present hypotheses the co-majority rule turns out to
be identical to the generalized Condorcet-Kemeny aggregation rule, namely a
rule that, at each profile of total preorders, selects one of the total
preorders of the finite set of alternative social states having a minimum
sum of Kemeny distances\footnote{%
Recall that the Kemeny distance between two binary relations is the size of
the symmetric difference between them (see section 4 and Appendix C for the
relevant definitions and more details).} from the preorders of the given
profile.
\end{description}

Thus, the foregoing results \textit{jointly} address \textit{several}
related issues:\ the characterization of the \textit{entire class} of
strategy-proof aggregation rules in \textit{median} finite semilattices, a
specific characterization of `median' rules amongst them, and -as a
by--product- both a general characterization of strategy-proof preference
aggregation rules including social welfare functions, and a specific
characterization of the Condorcet-Kemeny aggregation rule in a fixed
population setting.

As it turns out, such issues have been previously considered in the
literature, but most typically from mutually `disconnected' perspectives.
For instance, aggregation rules in semilattices and lattices have been
studied in depth in seminal contributions mostly due to Monjardet and his
co-workers (see e.g. Monjardet (1990)), but with no reference to
strategy-proofness issues. Condorcet-Kemeny aggregation rules have been
characterized at least for the case of \textit{linear orders} and in a 
\textit{variable population }setting (see e.g. Young, Levenglick (1978)),
but again with no reference to strategy-proofness properties. By contrast,
Bossert, Sprumont (2014) \textit{does} consider strategy-proofness issues,
and also provides characterizations of \textit{some }strategy-proof
preference aggregation rules in a variable population setting, but focusses
in fact on \textit{restricted aggregation rules }which admit \textit{%
arbitrary} total preorders as possible outputs while being \textit{only}
defined on profiles of \textit{linear orders}. Most recently, Bonifacio, Mass%
\'{o} (2020) essentially characterizes the \textit{sub-class of anonymous
and unanimity-respecting aggregation rules} in arbitrary join-semilattices
which are \textit{strategy-proof} on `single-peaked' domains of \textit{%
total preorders} (according to a notion of `single-peakedness' that reflects
the structure of the underlying join-semilattice). Since the relevant
join-semilattice may not be median, however, `median' rules such as
Condorcet-Kemeny rules are in general not available, and a \textit{specific }%
application to the case of total preorders and consequently to classical
Arrowian social welfare functions is in fact out of reach in that general
framework \textit{unless} some further structure is adjoined.

On the contrary, focussing on the case of \textit{median }join-semilattices
makes it possible and natural to address \textit{jointly }all of the
previous issues: that is precisely what is done in the present work. The
application of its main results to two prominent examples of median
join-semilattices (namely, the median semilattice of total preorders) and
the distributive lattice of reflexive binary relations on a finite set)
provides a simple way out of Arrow's `impossibility theorem' that also
ensures strategy-proofness for a quite large class of \ anonymous and
idempotent social welfare functions. It consists in combining full retention
of transitivity and a \textit{basic} version of Pareto optimality for social
preferences with a considerable \textit{relaxation of Arrow's Independence
of Irrelevant Alternatives.}

The rest of the paper is organized as follows. Section 2 is devoted to the
basic definitions and preliminaries of our model. Section 3 includes the
main results. Section 4 presents the application of the foregoing results to
preference aggregation rules. Section 5 collects some concluding remarks.
Finally, Appendix A includes the proof of the main result of the present
work, Appendix B contains some useful auxiliary notions concerning median
join-semilattices, and Appendix C provides a detailed discussion of related
literature.

\section{Notation, definitions and preliminaries}

Let $N=\left \{ 1,...,n\right \} $ denote the finite population of voters,
with $n\geq 3$ in order to avoid tedious qualifications. The subsets of $N$
are also referred to as \textit{coalitions, }and $(\mathcal{P}(N),\subseteq
) $ denotes the partially ordered set (\emph{poset}) of coalitions induced
by set-inclusion preorder on the set $\mathcal{P}(N)$ of all possible
subsets of $N$. A set $F\subseteq \mathcal{P}(N)$ of coalitions such that,
for any $S\in F$ and any $T\subseteq N$, if $S\subseteq T$ then $T\in F$ is
called \textit{order filter. }The set of inclusion-minimal
elements/coalitions of $F $ are the \textit{basis }of the order filter $F$
and it is denoted as $F^{\min }$.

Let $X$ be an arbitrary nonempty \textit{finite} set of alternatives and $%
\leqslant $ a partial order, i.e. a \emph{reflexive}, \emph{transitive} and 
\emph{antisymmetric} binary relation on $X$. Let $\mathcal{X}=(X,\leqslant )$
be the corresponding partially ordered set on $X$. We denote by $\vee $ and $%
\wedge $ the \textit{least-upper-bound} (or join) and \textit{%
greatest-lower-bound} (or meet) binary \textit{partial }operations on $X$ as
induced by $\leqslant $, respectively and by $\vee Y$ and $\wedge Y$ the
least-upper-bound and greatest-lower-bound of $Y$ (whenever they exist), for
any $Y\subseteq X$. The \textit{order filter} of a partially ordered set $%
(X,\leqslant )$ is a set $Y\subseteq X$ such that, for any $y,z\in X$, if $%
y\leqslant z$ and $y\in Y$, then $z\in Y$. For any $x\in X$, we then denote
with $\uparrow x=\left \{ y\in X:x\leqslant y\right \} $ the \textit{%
principal order filter\ }generated by $x$. An element $x\in X$ is \textit{%
meet-irreducible (join-irreducible) }if for any $Y\subseteq X$, $x=\wedge Y$
entails $x\in Y$ ($x=\vee Y$ entails $x\in Y$). Moreover, for any $%
Y\subseteq X$, $\wedge Y$ ($\vee Y$, respectively) is \textit{well-defined}
if and only if there exists $z\in X$ such that $z\leqslant y$ ($y\leqslant z$%
, respectively) for all $y\in Y$, namely the elements of $Y$ have a \textit{%
common lower (upper) bound}. The set of all meet-irreducible elements
(join-irreducible elements) of $\mathcal{X}=(X,\leqslant )$ will be denoted
by $M_{\mathcal{X}}$ ($J_{\mathcal{X}}$, respectively). Notice that, by
construction, for every $x\in X$, $x=\wedge M(x)$ where $M(x):=\left \{ m\in
M_{\mathcal{X}}:x\leqslant m\right \} $ and, dually $x=\vee J(x)$ where $%
J(x):=\left \{ j\in J_{\mathcal{X}}:j\leqslant x\right \} $.

The partially ordered set $\mathcal{X}=(X,\leqslant )$ is said to be a
(finite)\textit{\ join-semilattice (meet-semilattice}, respectively) if and
only if the least upper bound or joint $x\vee y$ (the greatest lower bound
or meet $x\wedge y$) is well-defined in $X$ for all $x,y\in X$, so that $%
\vee :X\times X\rightarrow X$ ($\wedge :X\times X\rightarrow X$) is a
function.\footnote{%
From now on, in order to avoid tedious repetitions and if any ambiguity is
excluded, we will basically only define one of the two main notions of
(joint or meet) semilattice used in the paper, assuming that the other can
simply be obtained by duality.} $\mathcal{X}=(X,\leqslant )$ is a \textit{%
lattice }if it is both a join-semilattice and a meet-semilattice. Notice
that a finite join-semilattice $\mathcal{X}=(X,\leqslant )$ has a (unique)
universal upper bound or \textit{top element} $\mathbf{1}=\vee X$, and its 
\textit{co-atoms} are those elements $x\in X$ such that $x\ll \mathbf{1}$,
with $\ll $ denoting the so-called \emph{cover relation}, meaning that there
is no $y\neq x$ such that $x\leqslant y\leqslant \mathbf{1}$ (see Appendix
C). The set of co-atoms of $\mathcal{X}=(X,\leqslant )$ is denoted by $%
\mathcal{C}_{\mathcal{X}}$. Dually, a finite meet-semilattice $\mathcal{X}%
=(X,\leqslant )$ has a (unique) universal lower bound or \textit{bottom
element} $\mathbf{0}=\wedge X$, and its \textit{atoms} are those elements $%
x\in X$ such that $\mathbf{0}\ll x$. The set of atoms of $\mathcal{X}$ is
denoted by $\mathcal{A}_{\mathcal{X}}$. Notice that a co-atom (atom,
respectively) is also a meet-irreducible (join-irreducible, respectively)
element. When co-atoms and meet-irreducibles (atoms and join-irreducibles)
do in fact \textit{coincide} the join-semilattice (meet-semilattice) is said
to be \textit{coatomistic} (\textit{atomistic}, respectively).

Let us now introduce the class of finite join-semilattices which is the
focus of the present paper.\medskip

\textbf{Definition 1.} A (finite) join-semilattice $\mathcal{X}=(X,\leqslant
)$ is \emph{median\ }if it satisfies the following pair of
conditions:\smallskip

\begin{description}
\item[$\left( i\right) $] \emph{upper distributivity}\textit{:} for all $%
u\in X$, and for all $x,y,z\in X$ such that $u$ is a lower bound of $%
\left
\{ x,y,z\right \} $, $x\vee (y\wedge z)=(x\vee y)\wedge (x\vee z)$
(or, equivalently, $x\wedge (y\vee z)=(x\wedge y)\vee (x\wedge z)$) holds.
Namely, $(\uparrow u,\leqslant _{|\uparrow u})$, where $\leqslant
_{|\uparrow u}$denotes the restriction of $\leqslant $ to $\uparrow u$, is a 
\emph{distributive lattice}\textit{;\footnote{%
We recall that a poset $(Y,\leqslant )$ is a \textit{distributive lattice}
if and only if, for any $x,y,z\in X$ , $x\wedge y$ and $x\vee y$ exist, and $%
x\wedge (y\vee z)=(x\wedge y)\vee (x\wedge z)$ (or, equivalently, $x\vee
(y\wedge z)=(x\vee y)\wedge (x\vee z)$). Moreover, a (distributive) lattice $%
\mathcal{X}$ is said to be \textit{lower (upper) bounded }if there exists $%
\bot \in X$ ($\top \in X)$ such that $\bot \leqslant x$ ($x\leqslant \top $)
for all $x\in X$, and \textit{bounded}\textbf{\ }if it is both lower bounded
and upper bounded.}}\smallskip

\item[$\left( ii\right) $] \emph{co-coronation}\textit{\ (\emph{or meet-Helly%
}): }for all $x,y,z\in X$ if $x\wedge y$, $y\wedge z$ and $x\wedge z$ exist,
then $(x\wedge y\wedge z)$ also exists\textit{.\medskip }
\end{description}

A well-known property of (finite) upper distributive join-semilattices that
will be used below is in the following:

\begin{claim}
Let $m\in M_{\mathcal{X}}$ be a meet-irreducible element of an upper
distributive finite join-semilattice $\mathcal{X}=(X,\leqslant )$ and $%
Y\subseteq X$ such that $\wedge Y$ exists. If $\wedge Y<m$ then there also
exists some $y\in Y$ such that $y\leqslant m$ (see e.g. Monjardet (1990)).
\end{claim}

It is easily checked that if $\mathcal{X}=(X,\leqslant )$ is a median
join-semilattice, then the partial function $\mu :X^{3}\rightarrow X$
defined as follows: for all $x,y,z\in X$,%
\begin{equation*}
\mu (x,y,z)=(x\vee y)\wedge (y\vee z)\wedge (x\vee z)
\end{equation*}%
is indeed a \textit{well-defined ternary operation }on $X$\textit{, }the%
\textbf{\ median}\textit{\ }of $\mathcal{X}$ which satisfies the following
two characteristic properties (see Sholander (1952, 1954):\smallskip

$(\mathbf{\mu }_{1})$ $\mu (x,x,y)=x$ for all $x,y\in X$

$(\mathbf{\mu }_{2})$ $\mu (\mu (x,y,v),\mu (x,y,w),z)=\mu (\mu (v,w,z),x,y)$
for all $x,y,v,w,z\in X$.\smallskip

Relying on $\mu $, we define a \textit{ternary }(\textit{median-induced}) 
\textbf{betweenness }relation:%
\begin{equation*}
B_{\mathcal{\mu }}=\left \{ (x,z,y)\in X^{3}:z=\mu (x,y,z)\right \}
\end{equation*}%
on $\mathcal{X}$, and, for any $x,y\in X$, the \textit{interval }induced by $%
x$ and $y$, namely:%
\begin{equation*}
I^{\mu }(x,y):=B_{\mu }(x,.,y)=\left \{ z\in X:z=\mu (x,y,z)\right \} \text{.%
}
\end{equation*}%
Therefore,\ for any $x,y,z\in X$, $(x,z,y)\in B_{\mathcal{\mu }}$ (also
written\textit{\ }$B_{\mathcal{\mu }}(x,z,y)$) if and only if\ $z\in I^{\mu
}(x,y)$. We recall here that the most appropriate interpretation of the
betweeness relation consists of considereing an `outcome $z$ that lies
between outcomes $x$ and $y$' as a `natural compromise' between $x$ and $y$,
namely, the betweenness relation is meant to represent a shared structure of
compromises between outcomes.

Furthermore, a (finite) median join-semilattice $\mathcal{X}=(X,\leqslant )$
admits a \textit{rank-based metric }$d_{r}:X\times X\rightarrow \mathbb{Z}%
_{+}$,\footnote{%
We recall that a \textit{metric} on $X$ is a real-valued function $\delta
:X\times X\rightarrow \mathbb{R}_{+}$ such that for any $x,y,z\in X$: (i) $%
\delta (x,y)=0$ iff $x=y$; (ii) $\delta (x,y)=\delta (y,x)$; (iii) $\delta
(x,z)\leq \delta (x,y)+\delta (y,z)$.} (with $r:X\rightarrow \mathbb{Z}_{+}$
denoting a \textit{rank function} (see Appendix C for the definition)),
defined, for any $x,y\in X$, as $d_{r}(x,y)=2r(x\vee y)-r(x)-r(y)$.\footnote{%
The metric median induced by this metric is strictly related to some
aggregation rules to be discussed in the present work. More details on this
topis are provided in Appendix C.}

Let $\succcurlyeq $ be a preorder, namely a reflexive and transitive binary
relation, on $X$,\footnote{%
We denote with $\succ $ and $\sim $ the asymmetric and symmetric components
of $\succcurlyeq $, respectively.} then we denote with $Top(\succcurlyeq )$
the possibly empty set of its maxima, and with $||$ the set of its \textit{%
incomparable }ordered pairs, i.e. $x||y$ if and only if neither $%
x\succcurlyeq y$ nor $y\succcurlyeq x$ hold. Then, we say that:\medskip

\textbf{Definition 2. }$\succcurlyeq $ is \textbf{locally unimodal\ }with
respect to the betweenness relation $B_{\mu }$, or $B_{\mu }$-\textbf{lu},
if and only if:

\begin{description}
\item[$(i)$] there exists a \textit{unique maximum} of $\succcurlyeq $ in $X$%
, its \textit{top }outcome, denoted $top(\succcurlyeq )$, and\smallskip

\item[$(ii)$] for all $x,y,z\in X$, if $z\in I^{\mu }(top(\succcurlyeq
),y)\backslash \left \{ top(\succcurlyeq )\right \} $ then $\left( not\text{ 
}y\succ z\right) $.\smallskip
\end{description}

The local unimodality of $\succcurlyeq $ amounts to the requirement that
individual preference relations have a unique maximum or top outcome and be
such that an outcome located between the maximum and another distinct
outcome is invariably regarded as not worse than the latter. Thus, \textit{%
local unimodality} is in fact nothing else than the specific notion of 
\textbf{single-peakedness }we are going to use in the present work.

We denote by $U_{B_{\mu }}$ the set of all $B_{\mu }$-\textbf{lu} preorders
on $X$ and by $U_{B_{\mu }}^{N}$ the set of all $N$-profiles of $B_{\mu }$-%
\textbf{lu} preorders, where an $N$-profile of $B_{\mu }$-\textbf{lu}
preorders is a mapping from $N$ into $U_{B_{\mu }}$. Moreover, we call a set 
$D_{\mathcal{X}}\subseteq U_{B_{\mu }}^{N}$ of locally unimodal preorders
with respect to $B_{\mu }$ \textbf{rich}\textit{\ }if for all $x,y\in X$
there exists $\succcurlyeq \in D_{\mathcal{X}}$ such that $top(\succcurlyeq
)=x$ and $UC(\succeq ,y)=I^{\mu }(x,y)$ (where $UC(\succeq ,y):=\left \{
y\in X:x\succcurlyeq y\right \} $ is the upper contour of $\succcurlyeq $ at 
$y$).

An \textbf{aggregation rule} for $(N,X)$ is a function $f:X^{N}\rightarrow X$%
. We (occasionally) consider both weaker and stricter versions of
aggregation rules, namely:\smallskip

$\left( i\right) $ a \textit{restricted aggregation rule} for $(N,X)$ is a
function $f:D\rightarrow X$ for some $D\subseteq X^{N}$;\smallskip

$\left( ii\right) $ a \textit{multi-aggregation rule }for $(N,X)$ is a
function $f:X^{N}\rightarrow \mathcal{P}(X)\setminus \left \{ \emptyset
\right \} $.\footnote{%
Notice that a multi-aggregation aggregation rule for $(N,X)$ can also be
regarded as an instance of a \textit{restricted aggregation rule }for $(N,%
\mathcal{P}(X))$.}\smallskip

$\left( iii\right) $ (by contrast), a \textit{constrained aggregation rule}
for $(N,X)$ is a function $f:X^{N}\rightarrow C$ for some $C\subseteq X$%
.\smallskip

All of the aforementiond notions of aggregation rule have been considered in
the relevant literature. In the present paper we shall focus on `pure'
aggregation rules (and occasionally on \textit{constrained }ones).

To proceed we define two compelling conditions on aggregation rules for
median joint-semi lattice which will play a crucial role in our main
characterization result.\smallskip

\textbf{Definition 3.} An aggregation rule $f:X^{N}\rightarrow X$ is said to
be:\smallskip

$\left( i\right) $ \textbf{strategy-proof }\textit{on }$U_{B_{\mu }}^{N}$ if
and only if, for all $B_{\mu }$-unimodal $N$-profiles $(\succcurlyeq
_{i})_{i\in N}\in $ $U_{B_{\mu }}^{N}$, and for all $i\in N$, $y_{i}\in X$,
and $(x_{j})_{j\in N}\in X^{N}$ such that $x_{j}=top(\succcurlyeq _{j})$ for
each $j\in N$, \textit{not }$f((y_{i},(x_{j})_{j\in N\smallsetminus \left \{
i\right \} }))\succ _{i}f((x_{j})_{j\in N})$;\smallskip

$\left( ii\right) $ $B_{\mu }$-\textbf{monotonic} if and only if, for all $%
i\in N$, $y_{i}\in X$, and $(x_{j})_{j\in N}\in X^{N}$,%
\begin{equation*}
f((x_{j})_{j\in N})\in I^{\mu }(x_{i},f(y_{i},(x_{j})_{j\in N\backslash
\left \{ i\right \} }))\text{.\footnote{$B_{\mu }$-monotonicity (or,
equivalently, $\mathcal{I}^{\mu }$-monotonicity) of $f$ amounts to requiring
all of its projections $f_{i}$ to be \textit{gate maps }to the image of $f$
(see van de Vel (1993), p.98 for a definition of gate maps). The
introduction of $B_{\mu }$-monotonic functions in a strategic social choice
setting is essentially due to Danilov (1994).}}
\end{equation*}%
\smallskip 

Indeed, a reliable and effective decision protocol should be reputedly 
\textit{strategy-proof}, i.e. immune to advantageous individual
manipulations through submission of false information, and $B_{\mu }$\emph{%
-monotonic}, i.e. the outcome that an agent obtains by submitting a certain
outcome $x$ lies between $x$ itself and the outcome that the agent would
obtain by submitting another outcome (for any fixed profile of
proposals/submissions on the part of the other agents). Thus, both
strategy-proofness on single-peaked domains and $B_{\mu }$-monotonicity of
an aggregation rule defined on a median semilattice are properties that are
by construction strictly related to the median-induced betweenness of the
semilattice (it will be shown below that they are in fact equivalent).

We further observe that non-trivial strategy-proof aggregation rules should
be, at least to some extent, \textit{input-responsive }and both \textit{%
input-unbiased} and \textit{output-unbiased. }A few requirements can be
deployed to present several versions, degrees and combinations of
input-responsiveness, input-unbiasedness and output-unbiasedness of
aggregation rules, namely:\smallskip

\textbf{Definition 4.} An aggregation rule $f$ for $(N,X)$ is:

\begin{itemize}
\item \textbf{inclusive }if and only if, for each voter $i\in N$, there
exist $x^{N}\in X^{N}$ and $y_{i}\in X$ such that $f(x^{N\backslash \left \{
i\right \} },y_{i})\neq f(x^{N})$;

\item \textbf{anonymous }if, for each $x^{N}\in X^{N}$ and each permutation $%
\sigma $ of $N$, $f(x^{N})=f(x^{\sigma (N)})$ (where $x^{\sigma
(N)}=(x_{\sigma (1)},...,x_{\sigma (n)})$);

\item \textbf{idempotent }(or \textit{unanimity-respecting) }if\ $%
f(x,...,x)=x$ for each $x\in X$;

\item \textbf{sovereign }if, for each $y\in X$, there exists $x^{N}\in X^{N}$
such that $f(x^{N})=y$ i.e. $f$ is an \textit{onto }function;

\item \textbf{neutral }if, for each $x^{N}\in X^{N}$ and each permutation $%
\pi $ of $X,$ $f(\pi (x^{N}))=\pi (f(x^{N}))$ (where $\pi (x^{N})=(\pi
(x_{1}),...,\pi (x_{k})))$.\footnote{%
Notice that both Idempotence and Neutrality imply Sovereignty (but not
conversely), while Anonymity and Sovereignty jointly imply Inclusiveness
(but not conversely). However, it is easily checked that if
Strategy-proofness holds, Sovereignty and Idempotence are in fact equivalent.%
}\medskip 
\end{itemize}

Anonimity, Idempotence and Neutral have a straighforward standard meaning.
Inclusiveness means that for each agent there exists at least one profile of
outcomes at which her own outcome turns out to be pivotal. Sovereignity
entails that every element of the aggregation rule's codomain, i.e. each
possible outcome, is the image of at least one element of its domain, i.e.
of each possible outcome profile.

In particular, let $\mathcal{X}=(X,\leqslant )$ be a finite \textit{%
join-semilattice} and $M_{\mathcal{X}}$ the set of its \emph{meet-irreducible%
} elements, and for any $x^{N}\in X^{N}$, and any $m\in M_{\mathcal{X}}$,
posit $N_{m}(x^{N}):=\left \{ i\in N:x_{i}\leqslant m\right \} $. Then, the
following properties of an aggregation rule can also be introduced:\smallskip

\textbf{Definition 5.} An aggregation rule $f:X^{N}\rightarrow X$ is:

$\left( i\right) $ $M_{\mathcal{X}}$\textbf{-independent }if and only if,
for all $x_{N},y_{N}\in X^{N}$ and all $m\in M_{\mathcal{X}}$, if $%
N_{m}(x_{N})=N_{m}(y_{N})$ then $f(x_{N})\leqslant m$ if and only if $%
f(y_{N})\leqslant m$;\smallskip

$\left( ii\right) \mathbf{\ }$\textbf{Isotonic} if $f(x_{N})\leqslant
f(x_{N}^{\prime })$ for all $x_{N},x_{N}^{\prime }\in X^{N}$ such that $x_{N}%
\mathbf{\leqslant }x_{N}^{\prime }$ (i.e. $x_{i}\leqslant x_{i}^{\prime }$
for each $i\in N$).\smallskip

Thus, an $M_{\mathcal{X}}$\textbf{-}independent\textbf{\ }rule ensures that
at any pair of profiles having the same set of agents proposing an outcome
consistent (respectively, inconsistent) with a certain join-irreducible
element, the social outcome will also be either consistent or inconsistent
with the latter in both cases. An aggregation rule $f$ is isotonic if it is
an order-preserving function. It can be easily shown (see Monjardet (1990))
that the \textit{conjunction} of $M_{\mathcal{X}}$\textbf{-Independence }and 
\textbf{Isotony }is equivalent to the following condition:\smallskip

\textbf{Definition 6.} An aggregation rule $f:X^{N}\rightarrow X$ is \textbf{%
monotonically }$M_{\mathcal{X}}$-\textbf{independent }if and only if, for
all $x_{N},y_{N}\in X^{N}$ and all $m\in M_{\mathcal{X}}$, if $%
N_{m}(x_{N})\subseteq N_{m}(y_{N})$ then $f(x_{N})\leqslant m$ implies $%
f(y_{N})\leqslant m$.\footnote{%
The notions of $J_{\mathcal{X}}$-Independence and Monotonic $J_{\mathcal{X}}$%
-Independence are defined similarly by dualization for a finite median
inf-semilattice $\mathcal{X}=(X,\leqslant )$ as follows: for all $%
x_{N},y_{N}\in X^{N}$ and all $j\in J_{\mathcal{X}}$, if $%
N_{j}(x_{N}):=\left \{ i\in N:j\leqslant x_{i}\right \} \subseteq
N_{j}(y_{N}):=\left \{ i\in N:j\leqslant y_{i}\right \} $, then $j\leqslant
f(x_{N})$ implies $j\leqslant f(y_{N})$.}\smallskip 

It should be noticed that $M_{\mathcal{X}}$\textbf{-Independence }amounts to
a weakening of the \textit{Arrovian Independence of Irrelevant Alternatives}
(more on this in Example 3 of Section 4).

\section{Main results}

We are now ready to state the main result of this paper concerning
strategy-proofness of aggregation rules on rich domains of locally unimodal
profiles in median join-semilattices.\footnote{%
A similar result holds for finite median meet-semilattices, and can be
easily established by dualization of the relevant arguments.}\smallskip

\textbf{Theorem 1.} Let $\mathcal{X}=(X,\leqslant )$ be a finite \textit{%
median join-semilattice, }$B_{\mu }\mathcal{\ }$its median-induced
betweenness, and $f:X^{N}\rightarrow X$ an aggregation rule for $(N,X)$.
Then, the following statements are equivalent:\smallskip

\ $(i)$ $f$ is strategy-proof on $D^{N}$ for any rich domain $D\subseteq
U_{B_{\mu }}$ of locally unimodal preorders w.r.t. $B_{\mu }$ on $X$%
;\smallskip

$(ii)$ $f$ is $B_{\mu }$-monotonic\textbf{;\smallskip }

$(iii)$ $f$ is monotonically $M_{\mathcal{X}}$-independent.\smallskip

A similar argument is used for the case of \textit{total} preorders on (not
necessarily finite) bounded \textit{distributive lattices} in Savaglio and
Vannucci (2019), and in Vannucci (2019). It should also be emphasized here
that, obviously, (finite) distributive lattices are a prominent special
subclass of (finite) median join-semilattices.

As a consequence of Theorem 1, we obtain the following:\smallskip

\textbf{Corollary 1.} Let $\mathcal{X}=(X,\leqslant )$ be a finite \textit{%
median join-semilattice, }$B_{\mu }\mathcal{\ }$its median-induced
betweenness, and \ $f:X^{N}\rightarrow X$ \ an aggregation rule. Then, the
following statements are equivalent:\smallskip

$(i)$ $f$ is strategy-proof on $D^{N}$ for every rich domain $D\subseteq
U_{B_{\mu }}$ of locally unimodal preorders with respect to. $B_{\mu }$ on $%
X $;\smallskip

$(ii)$ for each $m\in M_{\mathcal{X}}$ there exists an order filter $F_{m}$
of $(\mathcal{P}(N),\subseteq )$ such that:%
\begin{equation*}
f(x_{N})=f_{\left \{ F_{m}:m\in M_{\mathcal{X}}\right \}
}(x_{N}):=\tbigwedge \left \{ m\in M_{\mathcal{X}}:N_{m}(x_{N})\in
F_{m}\right \}
\end{equation*}%
for all $x_{N}\in X^{N}$.\smallskip

\begin{proof}
Immediate from Theorem 1 and dualization of Proposition 1.4 of Monjardet
(1990). In particular, each order filter $F_{m}$ consists of the \textit{%
locally }$m$\textit{-winning coalitions for }$f$\textit{,} namely for every $%
m\in M_{\mathcal{X}}$,%
\begin{equation*}
F_{m}:=\left \{ 
\begin{array}{c}
T\subseteq N:\text{there exists }x_{N}\in X^{N}\text{ such that } \\ 
\left \{ i\in N:x_{i}\leqslant m\right \} =T\text{ and }f(x_{N})\leqslant m%
\end{array}%
\right \} \text{.}
\end{equation*}
\end{proof}

It should be emphasized that the class of aggregation rules $f_{\left \{
F_{m}:m\in M_{\mathcal{X}}\right \} }$ identified by Corollary 1 is in
principle very comprehensive indeed. More specifically, Corollary $1.(ii)$
allows a broad description of such rules as those returning the \textit{%
strictest consensus among the admissible alternatives actually sponsored by
the agents of the relevant coalitions} (as specified by the order filters $%
F_{m}$). In particular, the class of aggregation rules thus characterized
encompasses a lot of suitably `inclusive' and/or `unbiased' rules, including
the following:

\begin{itemize}
\item \textit{Quorum system aggregation rules}, namely functions $%
f_{\left
\{ F_{m}:m\in M_{\mathcal{X}}\right \} }$ such that every order
filter $F_{m} $ is \textit{transversal, }i.e. $S\cap T\neq \varnothing $ for
all $S,T\in F_{m}$.

\item \textit{Inclusive aggregation rules, }namely functions $f_{\left \{
F_{m}:m\in M_{\mathcal{X}}\right \} }$ such that $\dbigcup \limits_{m\in M_{%
\mathcal{X}}}F_{m}^{\min }=N$.

\item \textit{Collegial aggregation rules, }namely functions $f_{\left \{
F_{m}:m\in M_{\mathcal{X}}\right \} }$ such that for some $m\in M_{\mathcal{X%
}}$, there exists a \textit{non-empty} $S_{m}\subseteq N$ with $%
F_{m}\subseteq \left \{ T\subseteq N:S_{m}\subseteq T\right \} $.

\item \textit{Outcome-biased aggregation rules}, namely functions $%
f_{\left
\{ F_{m}:m\in M_{\mathcal{X}}\right \} }$ where $F_{m}=\varnothing 
$ for some $m\in M_{\mathcal{X}}$.

\item \textit{Weakly neutral (or }$M_{\mathcal{X}}$-\textit{neutral)} 
\textit{aggregation rules}, namely functions $f_{\left \{ F_{m}:m\in M_{%
\mathcal{X}}\right \} }$ where $F_{m}=F_{m^{\prime }}$ whenever $m\wedge
m^{\prime }$ exists.

\item \textit{Quota aggregation rules}, namely anonymous aggregation rules
i.e. functions $f_{\left \{ F_{m}:m\in M_{\mathcal{X}}\right \} }$ such that
for each $m\in M_{\mathcal{X}}$ there exists an integer $q_{[m]}\leq |N|$
with $F_{m}=\left \{ T\subseteq N:q_{[m]}\leq |T|\right \} $. In particular, a
quota aggregation rule is also $M_{\mathcal{X}}-$\textit{neutral} (or weakly
neutral) if and only if $q_{[m]}=q_{[m^{\prime }]}$ whenever $m\wedge
m^{\prime }$ exists.
\end{itemize}

A prominent instance of \textit{a weakly neutral quota aggregation rule is} 
\textit{co-majority} as defined below.\textit{\smallskip }

\textbf{Definition 7.} (\textit{Co-majority rule) }Let $\mathcal{X}%
=(X,\leqslant )$ be a finite \textit{median join-semilattice, and }$N$ a
finite set. Then, the \textbf{co-majority rule} $f^{\partial maj}$\ for $%
(N,X)$ is defined as follows: for all $x_{N}\in X^{N}$,%
\begin{equation*}
f^{\partial maj}(x_{N}):=\dbigwedge \limits_{S\in \mathcal{W}%
^{maj}}(\dbigvee \limits_{i\in S}x_{i})
\end{equation*}%
where $W^{maj}=\left \{ S\subseteq N:|S|\geq \lfloor \frac{|N|+2}{2}\rfloor
\right \} $.\smallskip

It is easily seen, and left to the reader to check, that the co-majority
rule is in particular a positive instance of an \textit{idempotent,} \textit{%
inclusive} and \textit{transversal }aggregation rule.

As a further corollary of Theorem 1 and Corollary 1 we obtain a new
characterization of the co-majority rule via strategy-proofness, anonymity
as defined above and the following well-known general property for
aggregation rules, namely:\smallskip

\textbf{Definition 8.} An aggregation rule $f:X^{N}\rightarrow X$ is \textbf{%
Bi-Idempotent }whenever, for any $x_{N}\in X^{N}$ and $y,z\in X$, if $%
x_{i}\in \left \{ y,z\right \} $ for all $i\in N$, then $f(x^{N})\in
\left
\{ y,z\right \} $.\smallskip

Clearly, Bi-Idempotence amounts to a local requirement combining
`decisiveness' (the ability to select a single outcome) and `faithfulness'
(the ability to select the outcome among the proposals actually advanced)
both under perfect binary polarization and under perfect agreement.

Thus, we have the following characterization result of the co-majority
rule.\smallskip

\textbf{Proposition 1.} Let $\mathcal{X}$ $=(X,\leqslant )$ be a finite
median join-semilattice, $B_{\mathcal{\mu }}$ its median betweenness
relation, $D$ $\subseteq U_{B_{\mu }}$ a rich domain of locally unimodal
preorders with respect to $B_{\mathcal{\mu }}$. Then, an aggregation rule $%
f:X^{N}\rightarrow X$\ satisfies Anonymity, Bi-Idempotence and is
Strategy-proof on $D^{N}$ with $|N|$ odd if and only if $f$ is the
co-majority rule $f^{\partial maj}$.\smallskip

\begin{proof}
Immediate from Theorem 1 above and a straightforward dualization of
Corollary 7.4 of Monjardet (1990).
\end{proof}

The co-majority rule can also be seen as a way to compute a certain
metric-median of the outcome profiles to be aggregated. Of course, this also
applies to the case of total preorders to be discussed in the next section
(see Appendix C for more details on this important topic).

We are now ready to consider a most significant application of the previous
results that involves \textit{strategy-proof} \textit{aggregation of\
preferences} including (Arrowian) \textit{social welfare functions} and
their strategy-proofness properties\textit{: }the next section is entirely
devoted to that topic.

\section{Applications to strategy-proof preference aggregation}

The major examples of finite median join-semilattices we are going to
analyze involve the set of all total preorders on a finite set and consider
the corresponding (pure) aggregation rules, that are of course \textit{%
social welfare functions}. There are several but subtly distinct ways of
relating the set of total preorders to a median join-semilattice, regard
such collection as a subset of a larger collection of admissible preference
relations (e.g. reflexive and connected, or even just reflexive binary
relations). We focus on the first and more straightforward example, while
discussing only briefly the similar results that obtain from the application
of Theorem 1 and its Corollaries to the latter cases.\medskip

\textbf{Total preorders and social welfare functions.\smallskip }

\textbf{Example 1.} \textbf{The join-semilattice of total preorders on a
finite set.}

Let $A$ be a nonempty finite set of alternative social states, $\mathcal{R}%
_{A}^{T}$ the set of all total preorders (i.e. reflexive, transitive and
connected binary relations) on $A$. Let us define the join of two total
preorders on $A$ as the \textit{transitive closure} $\overline{\cup }$ of
their set-theoretic union. Then, by construction, $\mathcal{X}^{\prime }:=(%
\mathcal{R}_{A}^{T},\overline{\cup })$ is a join-semilattice, and satisfies
both \textit{upper distributivity} (by Claim (P.1) of Janowitz (1984)), and 
\textit{co-coronation} (by Claims (P.3) and (P.5) of Janowitz (1984)). It
follows that $(\mathcal{R}_{A}^{T},\overline{\cup })$ thus defined is indeed
a \textit{median join-semilattice }whose median ternary operation is denoted
here $\mu $, and its meet-irreducibles are the \textit{total preorders }$%
R_{A_{1}A_{2}}\in \mathcal{R}_{A}^{T}$ having just two (non-empty)
indifference classes $A_{1},A_{2}$ such that (i) $(A_{1},A_{2})$ is a
two-block ordered partition of $A$, written $(A_{1},A_{2})\in \Pi _{A}^{(2)}$%
, namely $A_{1}\cup A_{2}=A$, $A_{1}\cap A_{2}=\emptyset $ and (ii) [$%
xR_{A_{1}A_{2}}y$ and \textit{not }$yR_{A_{1}A_{2}}x$] if and only if $x\in
A_{1}$ and $y\in A_{2}$. Such total preorders $R_{A_{1}A_{2}}$ with $%
(A_{1},A_{2})\in \Pi _{A}^{(2)}$ are also the \textit{co-atoms of }$(%
\mathcal{R}_{A}^{T},\overline{\cup })$, hence, the join-semilattice of total
preorders is in particular \textit{co-atomistic.}

Thus, a most interesting application of our main result involves \textit{%
aggregation rules for preference profiles of total preorders }namely \textit{%
social welfare functions} $f:\mathcal{R}_{A}^{N}\rightarrow \mathcal{R}_{A}$
in the classic Arrowian sense.\footnote{%
Let $N,A$ \ be two (finite) sets and $\mathcal{R}_{A}^{T}$ the set of all
total preorders on $A$. An (\textit{Arrowian) social welfare function} for $%
(N,A)$ is a function $f:$ $(\mathcal{R}_{A})^{N}\rightarrow \mathcal{R}_{A}$%
, namely a function specifying a unique total preference preorder on $A$ for
every profile of $n$ total preference preorders on $A$. No further
requirement such as Independence of Irrelevant Alternatives or
Idempotence/Unanimity is assumed (see e.g. Sethuraman, Teo, Vohra (2003), or
Nehring, Puppe (2010) for the latter usage of the term). In the rest of this
paper Arrowian social welfare functions will be often referred to simply as 
\textit{`social welfare functions'. }Occasionally, and somewhat confusingly,
the very same label is also used to refer to (what we shall rather denote
as) \textit{Arrowian strict social welfare functions }$f:$ $(\mathcal{L}%
_{A})^{N}\rightarrow \mathcal{L}_{A}$ where $\mathcal{L}_{A}$ is the set of
all \textit{linear orders} (i.e. \textit{antisymmetric }total preorders on $%
A $). By contrast, a \textit{Bergson-Samuelson social welfare function }for $%
(N,A)$ is a function $f:\left \{ r^{N}\right \} \rightarrow \mathcal{R}%
_{A}^{T} $, with $r^{N}\in \mathcal{R}_{A}^{T}$.} Such an application is
made precise by the following proposition.\smallskip

\textbf{Proposition 2.} Let $A$ be a nonempty finite set of alternative
social states, $\mathcal{R}_{A}$ the set of all total preorders of $A$, $%
\mathcal{X}:=(\mathcal{R}_{A},\overline{\cup })$ the join-semilattice of
total preorders of $A$, $\mu $ the median ternary operation of $\mathcal{X}$%
, $B_{\mu }$ the corresponding betweenness relation as previously defined, $%
M_{\mathcal{X}}$ the set of all meet-irreducible elements of $\mathcal{X}$,
and $f:\mathcal{R}_{A}^{N}\rightarrow \mathcal{R}_{A}$ an aggregation rule
for $(N,\mathcal{R}_{A})$. Then, the following statements are equivalent:

$(i)$\ $f$ is strategy-proof on $D^{N}$ for every rich domain $D\subseteq
U_{B_{\mu }}$ of locally unimodal preorders w.r.t. $B_{\mu }$ on $\mathcal{R}%
_{A}$;

$(ii)$ for each $m\in M_{\mathcal{X}}$ there exists an order filter $F_{m}$
of $(\mathcal{P}(N),\subseteq )$ such that, for all $R_{N}\in \mathcal{R}%
_{A}^{N}$,%
\begin{equation*}
f(R_{N})=f_{\left \{ F_{m}:m\in M_{\mathcal{X}^{\prime }}\right \}
}(R_{N}):=\dbigcap \left \{ m\in M_{\mathcal{X}}:\left \{ i\in
N:R_{i}\subseteq m\right \} \in F_{m}\right \} \text{.}
\end{equation*}

\begin{proof}
Immediate, from Theorem 1 and Example 1.
\end{proof}

Notice that, as a consequence of the prevous characterization result, there
exist a large class of `classical' Arrowian social welfare functions on $%
(N,A)$ which are \textit{inclusive} and \textit{idempotent} (or
unanimity-respecting) as well as \textit{strategy-proof }on an arbitrary
rich domain of locally unimodal preorders with respect to the betweenness
relation $B_{\mu }$ of $(\mathcal{R}_{A}^{T},\overline{\cup })$ or $(%
\mathcal{T}_{A},\cup )$). Such a large class includes aggregation rules
which are respectively \textit{neither} anonymous nor neutral, \textit{just}
anonymous, \textit{just} neutral, or \textit{both} anonymous and neutral. To
see this, consider the following list of examples:

\begin{itemize}
\item \textit{Inclusive quorum systems, }namely functions $f_{\left \{
F_{m}:m\in M_{\mathcal{X}}\right \} }$ such that every order filter $F_{m}$
is \textit{transversal }i.e. $S\cap T\neq \varnothing $ for all $S,T\in
F_{m} $ and $\dbigcup \limits_{m\in M_{\mathcal{X}}}F_{m}^{\min }=N$.
Observe that such a\ class includes any rule such that for every $R_{m}\in
M_{\mathcal{X}} $, $F_{m}$ is \textit{simple-majority collegial }i.e. there
exists a \textit{minimal} simple majority coalition $S_{m}\subseteq N$, $%
|S_{m}|=\left \lfloor \frac{|N|+2}{2}\right \rfloor $ with $F_{m}=\left \{
T\subseteq N:S_{m}\subseteq T\right \} $. Generally speaking, inclusive
quorum systems need not be anonymous or neutral.

\item \textit{Outcome-biased aggregation rules}, namely functions $%
f_{\left
\{ F_{m}:m\in M_{\mathcal{X}}\right \} }$ where $F_{m}=\varnothing 
$ for some $m\in M_{\mathcal{X}}$. Observe that they include the subclass of
those aggregation rules such that for some \textit{total preorder }$%
\overline{R}\in \mathcal{R}_{A}^{T}$,\textit{\ }including possibly a\textit{%
\ linear order, }$F_{m}=\varnothing $ for every $m\in M_{\mathcal{X}}$ such
that $\overline{R}\subseteq m$.

\item (Weakly) \textit{Neutral aggregation rules}, namely functions $%
f_{\left \{ F_{m}:m\in M_{\mathcal{X}}\right \} }$ where $F_{m}=F_{m^{\prime
}} $ whenever $R_{m}\wedge R_{m^{\prime }}$ exists.

\item \textit{Quota aggregation rules}, namely functions $f_{\left \{
F_{m}:m\in M_{\mathcal{X}}\right \} }$ such that for each $m\in M_{\mathcal{X%
}}$ there exists an integer $q_{[m]}\leq |N|$ with $F_{m}=\left \{
T\subseteq N:q_{[m]}\leq |T|\right \} $. Such rules are clearly anonymous,
but not necessarily neutral: they are of course neutral as well if,
furthermore, $F_{m}=F_{m^{\prime }}$ whenever $m\cap m^{\prime }$
exists.\smallskip
\end{itemize}

It is worth noticing that a large subclass of such aggregation rules $f_{%
\mathcal{F}_{M_{\mathcal{X}}}}$, (including \textit{positive quota
aggregation rules} and \textit{inclusive quorum systems}) satisfy the 
\textit{Basic Pareto Principle} (BP), as made precise by the following
definition and claim.\smallskip

\textbf{Definition 9.} $\left( \emph{Basic\ Pareto\ Principle\ (BP)}\right) $%
\textbf{\ }An aggregation rule $f:\mathcal{R}^{N}\rightarrow \mathcal{R}$
with $\mathcal{R}\in \left \{ \mathcal{R}_{A}^{T},\mathcal{T}_{A}\right \} $
satisfies BP if for every $x,y\in A$ and $R_{N}\in \mathcal{R}_{A}^{N}$, if $%
xR_{i}y$ for every $i\in N$ then $xf(R_{N})y$.\smallskip

\begin{claim}
Let $\mathcal{R}\in \mathcal{R}_{A}^{T}$ and $\ f_{\mathcal{F}_{M_{\mathcal{X%
}}}}:\mathcal{R}^{N}\rightarrow \mathcal{R}$ be an aggregation rule as
defined above such that $F_{m}\ $is a nontrivial proper order filter (i.e. $%
\varnothing \notin F_{m}\neq \varnothing $) for every $m\in M_{\mathcal{X}}$%
. Then $f_{\mathcal{F}_{M_{\mathcal{X}}}}$ satisfies BP.
\end{claim}

\begin{proof}
Suppose that $x,y\in A$ and $R_{N}\in \mathcal{R}_{A}^{N}$ are such that $%
xR_{i}y$ for every $i\in N$\textit{, yet not }$xf_{\mathcal{F}_{M_{\mathcal{X%
}}}}y$. Namely, by construction,%
\begin{equation*}
(x,y)\notin \dbigcap \left \{ m\in M_{\mathcal{X}}:\left \{ i\in
N:R_{i}\subseteq m\right \} \in F_{m}\right \} \text{\textit{.}}
\end{equation*}

Hence, there exists\textit{\ }$m\in M_{\mathcal{X}}$ such that $\left \{
i\in N:R_{i}\subseteq m\right \} \in F_{m}$ and $(x,y)\notin m$. However, by
assumption , $F_{m}$ is nonempty and \textit{every }$T\in F_{m}$ is itself
nonempty: thus, $N\in F_{m}$. But then $(x,y)\in R_{i}\subseteq m$ for any $%
i\in T$, a contradiction.
\end{proof}

A remarkable family of \textit{anonymous} but typically \textit{not} neutral
aggregation rules for $(N,\mathcal{R})$ (with $\mathcal{R}\in \left \{ 
\mathcal{R}_{A}^{T},\mathcal{T}_{A}\right \} $ ) is that of \textit{%
Condorcet-Kemeny rules}, as defined below (see also Young, Levenglick
(1978), Young (1995)).\smallskip

\textbf{Definition 10. }\textit{(Generalized Condorcet-Kemeny aggregation
rules) }Let\textit{\ }$\mathcal{X}:=(\mathcal{R},\cup )\in (\mathcal{R}%
_{A}^{T},\overline{\cup })$ be the join-semilattice of total preorders on
finite set $A$, $C(\mathcal{X})$ its covering graph, $\delta _{C(\mathcal{X}%
)}$ the shortest-path metric on $C(\mathcal{X})$,\footnote{%
See Appendix B for a precise definition of the covering graph and the
shortest-path metric of $\mathcal{X}$.} $\mathcal{L}_{A}\subseteq \mathcal{R}
$ the set of linear orders on $A$, $N$ a finite set, and $\mathbf{\preceq }$
a linear order on $\mathcal{R}$.

The \textit{generalized Condorcet-Kemeny aggregation rule }for $(N,\mathcal{R%
})$ induced by $\mathbf{\preceq }$ is the function \textit{\ }$f_{\preceq
}^{CK^{\ast }}:\mathcal{R}^{N}\rightarrow \mathcal{R}$ defined as follows:
for all $R_{N}\in \mathcal{R}^{N}$,%
\begin{equation*}
f_{\preceq }^{CK^{\ast }}(R_{N}):=\min_{\preceq }\left \{ 
\begin{array}{c}
R\in \mathcal{R}:\dsum \limits_{i\in N}\delta _{C(\mathcal{X})}(R,R_{i})\leq
\dsum \limits_{i\in N}\delta _{C(\mathcal{X})}(R^{\prime },R_{i})\text{ \ }
\\ 
\text{for all }R^{\prime }\in \mathcal{R}%
\end{array}%
\right \} \text{.}
\end{equation*}%
In particular, the (strict) \textit{Condorcet-Kemeny aggregation rule }for $%
(N,\mathcal{R})$ induced by $\mathbf{\preceq }$ is the function \textit{\ }$%
f_{\preceq }^{CK}:\mathcal{R}^{N}\rightarrow \mathcal{L}_{A}$ defined as
follows: for all $R_{N}\in \mathcal{R}^{N}$,%
\begin{equation*}
f_{\preceq }^{CK}(R_{N}):=\min_{\preceq }\left \{ 
\begin{array}{c}
R\in \mathcal{L}_{A}:\dsum \limits_{i\in N}\delta _{C(\mathcal{X}%
)}(R,R_{i})\leq \dsum \limits_{i\in N}\delta _{C(\mathcal{X})}(R^{\prime
},R_{i})\text{ } \\ 
\text{\ for all }R^{\prime }\in \mathcal{L}_{A}%
\end{array}%
\right \} \text{.}
\end{equation*}%
\smallskip

Notice that a (strict) Condorcet-Kemeny rule amounts to a \textit{constrained%
} generalized Condorcet-Kemeny rule. It should also be emphasized that
generalized Condorcet-Kemeny aggregation rules require a prefixed linear
order $\preceq $ as a tie-breaker device whenever the \textit{remoteness} 
\textit{function} $\dsum \limits_{i\in N}\delta _{C(\mathcal{X})}(.,R_{i})$
of a profile $R_{N}$ admits several distinct minima: that is \textit{the
only role} of $\preceq $ in $f_{\preceq }^{CK}$ and $f_{\preceq }^{CK^{\ast
}}$, and the source of the typical failure of Condorcet-Kemeny rules to
satisfy Neutrality.\footnote{%
The computational complexity issues raised by computation of the
Condorcet-Kemeny aggregation rule will not be addressed in the present work.
However, it is worth mentioning here that the computation of median total
preorders for arbitrary profiles of total preorders is a NP-complete problem
(i.e. it belongs to the class of the hardest problems whose solutions are
polynomial-time verifiable or `easy' to verify, but apparently worst-case
`hard' to compute). Specifically, if the size of $N$ is suitably larger than
the size of $A$, computing a median total preorder is NP-complete for
arbitrary profiles of total preorders or linear orders (Hudry (2012)) and
NP-hard (i.e. `easy' to reduce to a NP-complete problem) for arbitrary
profiles of binary relations (Wakabayashi (1998)).} It follows that, to the
extent that uniqueness of minima of the remoteness function is warranted,
the outcome of Condorcet-Kemeny rules is unaffected by the choice of $%
\preceq $ and Neutrality is restored. That is precisely the case when the
size $n$ of the set of agents $N$ is odd, as implied by the following
characterization result:\smallskip

\textbf{Proposition 3.} Let $\mathcal{X}:=(\mathcal{R}_{A},\overline{\cup })$
be the join-semilattice of total preorders on finite set $A$ as defined
above, $\mu $ its median ternary operation and $B_{\mu }$ the corresponding
betweenness relation as previously defined, $N$ a finite set such that $|N|$
is an \textit{odd number, }and $f:\mathcal{R}_{A}^{N}\rightarrow \mathcal{R}%
_{A}$ an aggregation rule for $(N,\mathcal{R}_{A})$. Then, the following
statements are equivalent:

$(i)$\ $f$\ satisfies Anonymity and Bi-Idempotence, and is strategy-proof on 
$D^{N}$ for every rich domain $D\subseteq U_{B_{\mu ^{\prime }}}$ of locally
unimodal preorders w.r.t . $B_{\mu }$ on $\mathcal{R}_{A}$;\smallskip

$(ii)$ $f=$ $\widehat{f}^{\partial maj}$where $\widehat{f}^{\partial maj}:%
\mathcal{R}_{A}^{N}\rightarrow \mathcal{R}_{A}$ denotes the co-majority
aggregation rule for $(N,\mathcal{R}_{A})$;\smallskip

$(iii)$ $f=f_{\preceq }^{CK}=f_{\preceq ^{\prime }}^{CK}$, i.e. the
generalized Condorcet-Kemeny aggregation rule for $(N,\mathcal{R}_{A})$ for
any pair of linear orders $\preceq $, $\preceq ^{\prime }$ on $\mathcal{R}$.

\begin{proof}
Immediate from Theorem 1, Proposition 1, Claim 1, and Example 1 above.
\end{proof}

Thus, in particular, when the size of $N$ is odd the generalized
Condorcet-Kemeny rule for $(N,\mathcal{R}_{A})$ is precisely the same as the
co-majority rule, and can be characterized as the unique aggregation rule
for $(N,\mathcal{R}_{A})$ (or, in other terms, the unique \textit{Arrowian
social welfare function)} which is Anonymous, Bi-Idempotent and
strategy-proof on $U_{B_{\mu ^{\prime }}}$(and any of its rich subdomains).%
\footnote{%
Is should be noticed that the requirement that $n=|N|$ be odd is not at all
as restrictive as it might seem at first sight. In fact, for $n$ even our
aggregation rule $f$ \ for $(N,\mathcal{R}_{A}^{T})$ might be embedded in a
natural way into a more comprehensive aggregation rule $\widetilde{f}$ for $%
(N,\mathcal{R}_{A}^{T}\times \mathbb{Z})$ (where $\mathbb{Z}$ denotes the
set of integer numbers) as supplemented with the natural projection from $%
\mathbb{Z}$ to the finite additive group $\mathbb{Z}_{n}$ of integers modulo 
$n$. Such an aggregation rule implements a pseudo-random `anonymous'
selection of a `president' in $N$ to the effect of producing an artificially
but fairly augmented `electorate' of \textit{odd }size. Furthermore, a
similar construct obtained by replacing $\mathbb{Z}_{n}$ with $\mathbb{Z}%
_{k} $ (where $k:=|A|$) results in a further aggregation rule $%
\overrightarrow{f}$ for $(N,\mathcal{R}_{A}^{T}\times \mathbb{Z})$ which
implements a pseudorandom `neutral' choice of \textit{one} linear order
among those consistent with the total preorder selected by $f$ at any
profile. Such an aggregation rule $\overrightarrow{f}$ is in fact \textit{%
constrained }(actually an $\mathcal{L}_{A}^{T}$-constrained one), since its
values are constrained to lie in $\mathcal{L}_{A}^{T}\subseteq \mathcal{R}%
_{A}^{T}$.}

Moreover, notice that (for an odd $n$) $\widehat{f}^{\partial maj}\equiv
f_{\preceq }^{CK^{\ast }}$satisfies a weak version of the so-called \textbf{%
Condorcet principle}, namely for every $(R_{i})_{i\in N}\in \mathcal{R}_{A}$
and $x\in A$, \textit{if} $x$ is a \textit{Condorcet winner}, that is $%
\left
\{ i\in N:xR_{i}y\text{ and \textit{not }}yR_{i}x\right \} \in 
\mathcal{W}^{maj}$ for every $y\in A\setminus \left \{ x\right \} $, \textit{%
then} $x\in Top(\widehat{f}^{\partial maj}((R_{i})_{i\in N}))$, (where for
any $R\in \mathcal{R}_{A}$, $Top(R):=\left \{ x\in A:xRy\text{ for all }y\in
A\right \} $).

To check this, suppose $x$ is indeed a Condorcet winner, yet $x\notin Top(%
\widehat{f}^{\partial maj}((R_{i})_{i\in N}))$. Thus, there exist $y\in
X\setminus \left \{ x\right \} $ and a meet-irreducible $R_{[y][x]}$ of the
join-semilattice $(\mathcal{R}_{A},\overline{\cup })$, (i.e. a
two-indifference-class total preorder having $y$ among its maxima and $x$
among its minima), such that $Top(\widehat{f}^{\partial maj}((R_{i})_{i\in
N}))\subseteq R_{[y][x]}$. But then, upper distributivity of $(\mathcal{R}%
_{A},\overline{\cup })$ entails that $\overline{\tbigcup \limits_{i\in T}}%
R_{i}\subseteq R_{[y][x]}$ for some $T\in \mathcal{W}^{maj}$ whence $%
R_{i}\subseteq R_{[y][x]}$ for each $i\in T\in \mathcal{W}^{maj}$, a
contradiction.

It is easily checked that $\widehat{f}^{\partial maj}$ is also $M_{\mathcal{X%
}}$-Neutral if $n:=|N|$ is odd . It follows that for any odd $n$ there
exists an Arrowian social welfare function on the \textit{full} domain of
total preorders on an arbitrary finite set which is anonymous, neutral,
idempotent (because Bi-Idempotence clearly implies Idempotence), satisfies a
monotonic independence property w.r.t. the meet-irreducible total preorders
(which are the co-atoms of the join-semilattice $(\mathcal{R}_{A},\subseteq
) $, i.e. the total preorders having just two indifference classes) \textit{%
and is strategy-proof} on any rich locally unimodal preference domain on $%
\mathcal{R}_{A}$. Therefore, $\widehat{f}^{\partial maj}$ is in particular a 
\textit{social welfare function that satisfies all the properties required
by Arrow's (Im)Possibility Theorem except for} the\textit{\ Independence of
Irrelevant Alternatives (IIA) condition}. \footnote{%
Recall that Arrow's IIA (in binary form) is a condition on social welfare
functions $f$ :$(\mathcal{R}_{A}^{T})^{N}\rightarrow \mathcal{R}_{A}^{T}$
defined as follows: for every $x$,$y\in A$ and any $R_{N}$ $,R_{N}^{\prime
}\in (\mathcal{R}_{A}^{T})^{N}$ such that $xR_{i}y$ if and only if $%
xR_{i}^{\prime }y$ for each $i\in N$, $xf(R_{N})y$ entails $%
xf((R_{N}^{\prime })y$.}\textit{\ }What is then the relationship between $M_{%
\mathcal{X}}$-Independence ($M_{\mathcal{X}}$-$I$) and IIA? It is quite
clear that under Idempotence $M_{\mathcal{X}}$-$I$ is definitely \textit{%
weaker} than IIA because, as a consequence of Proposition 1, the former is
consistent with Anonymity and Neutrality of an (Arrowian)
unanimity-respecting social welfare function while the latter is not.
Indeed, as established by Hansson (1969), IIA in combination with Anonymity
and Neutrality provides a characterization of the \textit{constant }social
welfare function having the \textit{universal indifference} relation $%
A\times A$ as its \textit{unique} value (hence in particular the former
combination of properties is inconsistent with Idempotence). In other terms,
strenghtening $M_{\mathcal{X}}$-Independence to IIA is \textit{just
impossible} for unanimity-respecting, anonymous and neutral Arrowian social
welfare functions.

\begin{remark}
(\textbf{The case of weak orders}) A \textbf{weak order}\textit{\ }on $A$ is
a binary relation $W$ that satisfy \textit{asymmetry (}$xWy$ entails $not$ $%
yWx$ for every $x,y\in A$) and negative transitivity ($not$ $xWy$ and $not$ $%
yWz$ entail $not$ $xWz$ for all $x,y,z\in A$). It is easily checked that the
partially ordered set $(\mathcal{W}_{A},\subseteq ^{\partial })$ of weak
orders on $A$ (where $W\subseteq ^{\partial }W^{\prime }$ if and only if $%
W^{\prime }\subseteq W$) is isomorphic to $(\mathcal{R}_{A},\subseteq )$
hence a median join-semilattice. It follows that versions of \ Corollary 1
and Proposition 1 also hold true for weak orders on a finite set.
\end{remark}

As it turns out, reconciling unanimity-respecting and strategy-proof
preference aggregation to IIA is however possible by moving away from the
domain of total preorders, towards some more comprehensive preference
domains. This observation brings us to a few other examples, to which we now
turn.\footnote{%
A \textit{further }relevant example is the median join-semilattice $\mathcal{%
R}_{A}^{T}\times \mathcal{P}(A)$, which is particularly convenient when it
comes to addressing squarely agenda-manipulation issues. Such semilattice
will be discussed in some detail elsewhere.}\medskip

\textbf{Other types of preference relations and preference aggregation
rules.\smallskip }

To begin with, let us consider the collections of \textit{generalized weak
tournaments} and of \textit{generalized strict tournaments}, namely the set
of all reflexive (respectively, irreflexive) connected relations on a finite
set, no matter if transitive or not.

\bigskip

\textbf{Example 2.} \textbf{Two isomorphic join-semilattices: the
join-semilattices of generalized weak tournaments and of generalized strict
tournaments on a finite set.}

Let $A$ be a nonempty finite set of alternative social states with $|A|=m$, $%
\mathcal{T}_{A}$ the set of all \textit{generalized weak} \textit{tournaments%
} namely \textit{reflexive} and \textit{connected} binary relations on $A$.
Let us define the join of two total relations on $A$ as their set-theoretic
union $\cup $. Then, by construction, $\mathcal{X}^{\prime }:=(\mathcal{T}%
_{A},\cup )$ is a join-semilattice, and its meet-irreducibles are the $%
m\cdot (m-1)$ total relations $T_{-xy}:=\left \{ (a,b)\in A^{2}:(a,b)\neq
(x,y)\right \} $ with $x,y\in A,$ $x\neq y$: namely, total relations whose
asymmetric components consist of some \textit{single} ordered pair $(x,y)$.
It can also immediately checked that such $T_{-xy}$ total relations are
indeed the \textit{co-atoms of }$(\mathcal{T}_{A},\cup )$. \textit{It turns
out that validity of the following claim can be easily established}\textbf{. 
}The join-semilattice $\mathcal{X}^{\prime \prime }:=$ $(\mathcal{T}%
_{A}^{\circ },\cup )$ is defined similarly on the set $\mathcal{T}%
_{A}^{\circ }$ of all \textit{generalized strict} \textit{tournaments},
namely \textit{irreflexive} and \textit{connected} binary relations on $A$.
It is easily checked that $\mathcal{X}^{\prime }$ and $\mathcal{X}^{\prime
\prime }$ are isomorphic join-lattices, with an isomorphism $\psi :\mathcal{T%
}_{A}\rightarrow \mathcal{T}_{A}^{\circ }$ between them being defined by the
rule $\mathbb{\psi }(T):=T\setminus \Delta _{A}$ for any $T\in \mathcal{T}%
_{A}$, with $\Delta _{A}:=\left \{ (x,x):x\in A\right \} $.

\begin{claim}
\textbf{\ }For any nonempty finite set $A$, the join-semilattices $(\mathcal{%
T}_{A},\cup )$ and $(\mathcal{T}_{A}^{\circ },\cup )$ are median.\footnote{%
The proof is available from the authors upon request.}
\end{claim}

Clearly, versions of Propositions 2 and 3 (and of Claim 2) also hold true
for both generalized weak tournaments and generalized strict tournaments.

Next, we proceed to remove the connectedness requirement as well.\medskip

\textbf{Example 3.} \textbf{Two isomorphic lattices: the lattices of
reflexive binary relations and of irreflexive binary relations on a finite
set.}

Clearly enough, any distributive lattice $(X,\vee ,\wedge )$ also provides
an example of a median join-semilattice.

In particular, let $A$ be a nonempty finite set of alternative social
states, $B_{A}^{r}$ the set of all reflexive binary relations on $A$, $%
(B_{A}^{r},\subseteq )$ the set-inclusion poset on $B_{A}^{r}$. Let us then
define the join $\vee $ and meet $\wedge $ of two reflexive binary relations
on $A$ as their set-theoretic union $\cup $ and intersection $\cap $,
respectively. Hence, $\mathcal{X}^{\prime }:=(B_{A}^{r},\cup ,\cap )$ is
indeed, by construction, a (bounded) \textit{distributive lattice}. It
follows that $\cup $-closedness of $B_{A}^{r}$ and both upper-distributivity
and co-coronation trivially hold in $\mathcal{X}^{\prime }$, i.e. $%
(B_{A}^{r},\cup )$ is in particular a \textit{median join-semilattice} whose
median $\mu ^{\prime }$ is precisely the median of the distributive lattice $%
(B_{A}^{r},\cup ,\cap )$. Namely, for any $R_{1},R_{2},R_{3}\in B_{A}^{r}$,%
\begin{equation*}
\mu ^{\prime \prime }(R_{1},R_{2},R_{3})=(R_{1}\cup R_{2})\cap (R_{2}\cup
R_{3})\cap (R_{3}\cup R_{1})=(R_{1}\cap R_{2})\cup (R_{2}\cap R_{3})\cup
(R_{3}\cap R_{1}).
\end{equation*}

Moreover, it can be easily shown (and left to the reader to check) that%
\begin{equation*}
M_{\mathcal{X}^{\prime }}=\mathcal{C}_{\mathcal{X}^{\prime }}=\left \{
A^{2}\setminus \left \{ (a,b)\right \} :a,b\in A,a\neq b\right \} \text{,}
\end{equation*}
and%
\begin{equation*}
J_{\mathcal{X}^{\prime }}=\mathcal{A}_{\mathcal{X}^{\prime }}=\left \{
\Delta _{A}\cup \left \{ (a,b)\right \} :a,b\in A,a\neq b\right \} \text{,}
\end{equation*}%
where $\Delta _{A}:=\left \{ (a,a):a\in A\right \} $.

Hence, $\mathcal{X}^{\prime }$ is in particular a \textit{co-atomistic }and 
\textit{atomistic} lattice.

It should be emphasized that the set of all total preorders on $A$ is
clearly a \textit{subset, }but \textit{not} a sub-join semilattice of the
join-semilattice reduct $(B_{A}^{r},\cup )$ of the lattice $(B_{A}^{r},\cup
,\cap )$, since the union of two total preorders may \textit{not} be
transitive.\footnote{%
To see this, consider e.g. $A=\left \{ a,b,c,d\right \} $,\ and the linear
orders, $R_{1}:=abcd$, $R_{2}:=dbca$ (written according to the usual
`decreasing' notation). Now, $R_{1}\cup R_{2}=\left \{ (x,x):x\in A\right \}
\cup \left \{
(a,b),(b,a),(a,c),(c,a),(a,d),(d,a),(b,c),(b,d),(d,b),(c,d),(d,c)\right \} $
which is not transitive since $\left \{ (c,a),(a,b)\right \} \subseteq
R_{1}\cup R_{2}$ but $(c,b)\notin R_{1}\cup R_{2}$).}

Let us turn now to the set $B_{A}^{ir}$ the set of all \textit{irreflexive}
binary relations on $A$ (namely, the binary relations $R\subseteq A^{2}$
such that $(x,x)\notin R$ for every $x\in A$) and to $(B_{A}^{ir},\subseteq
) $, namely the set-inclusion poset on $B_{A}^{ir}$. Since both the
intersection and the union of two irreflexive relations are also
irreflexive, $\mathcal{X}^{\prime \prime }:=(B_{A}^{ir},\cup ,\cap )$ is
indeed, by construction, a (bounded) \textit{distributive lattice}.
Moreover, $(B_{A}^{r},\cup ,\cap )$ and $(B_{A}^{ir},\cup ,\cap )$ are
isomorphic lattices, an obvious isomorphism $\varphi :B_{A}^{r}\rightarrow
B_{A}^{ir}$ between them being defined by the rule $\varphi (R):=R\setminus
\Delta _{A}$.

Let us now introduce the strenghtening of \textit{Monotonic }$M_{\mathcal{X}%
} $\textit{-Independence }which results from substituting IIA for $M_{%
\mathcal{X}}$\textit{-}Independence.\smallskip

\textbf{Definition 10.} $\left( \emph{Monotonic\ IIA}\right) $\textbf{: }Let 
$A$ be a nonempty finite set of alternative social states, $\mathcal{B}%
_{A}\in $ $\left \{ \mathcal{B}_{A}^{r},\mathcal{B}_{A}^{ir}\right \} $, $%
\mathcal{X}^{\prime }:=(\mathcal{B}_{A},\cup ,\cap )\mathcal{\ }$the
(bounded) distributive lattice induced on $\mathcal{B}_{A}$ by $\cup $ and $%
\cap $ , and $f:(\mathcal{B}_{A})^{N}\rightarrow \mathcal{B}_{A}$ an
aggregation rule for $(N,\mathcal{B}_{A})$. Then, $f$ is \textbf{%
monotonically IIA} if, \ for all $R_{N},R_{N}^{\prime }\in (\mathcal{B}%
_{A})^{N}$ and all $(u,v)\in A^{2}$: if $\left \{ i\in N:uR_{i}v\right \}
\subseteq \left \{ i\in N:uR_{i}v\right \} $ then $u$ $f(x_{N})v$ implies $%
uf(y_{N})v$.\smallskip

We are now ready to show that when the relevant join-semilattice is (the
join-reduct of) $\mathcal{X}^{\prime }\in \left \{ (\mathcal{B}_{A}^{r},\cup
,\cap ),(\mathcal{B}_{A}^{ir},\cup ,\cap )\right \} $ we can rely on the
full force of a counterpart of Theorem 1 for bounded distributive lattices
(see e.g. Savaglio, Vannucci (2019)) to obtain the following
result.\smallskip

\textbf{Proposition 4.} Let $A$ be a nonempty finite set of alternative
social states, $\mathcal{B}_{A}\in $ $\left \{ \mathcal{B}_{A}^{r},\mathcal{B%
}_{A}^{ir}\right \} $, $\mathcal{X}^{\prime }:=(\mathcal{B}_{A},\cup ,\cap )%
\mathcal{\ }$the (bounded) distributive lattice induced on $\mathcal{B}_{A}$
by $\cup $ and $\cap $, $\mu ^{\prime }$ its median ternary operation and $%
B_{\mu ^{\prime }}$ the corresponding betweenness as previously defined, and 
$f:(\mathcal{B}_{A})^{N}\rightarrow \mathcal{B}_{A}$ an aggregation rule for 
$(N,\mathcal{B}_{A})$. Then, the following statements are equivalent:

$(i)$\ $f$ is strategy-proof on $D^{N}$ for every rich domain $D\subseteq
U_{B_{\mu ^{\prime \prime }}}$ of locally unimodal preorders w.r.t. $B_{\mu
^{\prime }}$ on $\mathcal{B}_{A}$;

$(ii)$ $f$ is $B_{\mu ^{\prime }}$-monotonic;

$(iii)$ $f$ is monotonically $M_{\mathcal{X}^{^{\prime }}}$-independent;

$(iv)$ $f$ is monotonically $J_{\mathcal{X}^{^{\prime }}}$-independent;

$(v)$ $f$ is monotonically IIA;

$(vi)$ there exists an order filter $\mathcal{F}$ of $(\mathcal{P}%
(N),\subseteq )$ and a family $\left \{ R_{S}\in \mathcal{B}_{A}:S\in 
\mathcal{F}\right \} $ of relations in $\mathcal{B}_{A}$ such that $%
f(R_{N})=\dbigcap \limits_{S\in \mathcal{F}}((\cup _{i\in S}R_{i})\cup
R_{S}) $ for all $R_{N}\in (\mathcal{B}_{A}^{r})^{N}$;

$(vii)$ there exists an order filter $\mathcal{F}$ of $(\mathcal{P}%
(N),\subseteq )$ and a family $\left \{ R_{S}\in \mathcal{B}_{A}:S\in 
\mathcal{F}\right \} $ of relations in $\mathcal{B}_{A}$ such that $%
f(R_{N})=\dbigcup \limits_{S\in \mathcal{F}}((\cap _{i\in S}R_{i})\cap
R_{S}) $ for all $R_{N}\in (\mathcal{B}_{A})^{N}$.\footnote{%
We omit the proof, which relies on a result mentioned in Davey, Priestley
(1990), p. 178, and is available from the authors upon request.}\smallskip

It goes without saying that the strategy-proof aggregation rules for $(N,%
\mathcal{B}_{A})$ characterized above (with $\mathcal{B}_{A}\in $ $\left \{ 
\mathcal{B}_{A}^{r},\mathcal{B}_{A}^{ir}\right \} $) comprise counterparts
to inclusive quorum systems, quota rules and all the other aggregation rules
for total preorders mentioned above. The co-majority rule $f_{r}^{\partial
maj}:(\mathcal{B}_{A})^{N}\rightarrow \mathcal{B}_{A}$ is defined by the
identity $f_{r}^{\partial maj}(R_{N})=\dbigcap \limits_{S\in W^{maj}}(\cup
_{i\in S}R_{i})$ for each $R_{N}\in (\mathcal{B}_{A})^{N}$, which is
obtained from the general formula under statement (vi) of the previous
proposition by setting $\mathcal{F}=W^{maj}:=\left \{ S\subseteq N:|S|\geq
\lfloor \frac{|N|+2}{2}\rfloor \right \} $ and $R_{S}=\Delta _{A}$ for each $%
S\in W^{maj}$.

But new possibilities arise here. To begin with, a version of the \textit{%
majority rule} $f_{r}^{maj}:(\mathcal{B}_{A})^{N}\rightarrow \mathcal{B}_{A}$
is now well-defined by the identity $f_{r}^{maj}(R_{N})=\dbigcup
\limits_{S\in W^{maj}}(\cap _{i\in S}R_{i})$ for each $R_{N}\in (\mathcal{B}%
_{A})^{N}$, which is obtained from the general formula under statement (vii)
of the previous proposition by setting $\mathcal{F}=W^{maj}$ and $%
R_{S}=A\times A$ for each $S\in W^{maj}$.

Of course, the outputs of $f_{r}^{maj}$(and $f_{r}^{\partial maj}$, or for
that matter of any idempotent aggregation rule for $(N,\mathcal{B}_{A})$)
may well be nontransitive or even \textit{intransitive }(i.e. include cycles
with asymmetric components). To see this, just consider a profile consisting
of identical \textit{cyclic }reflexive relations. But then, it is easily
seen (and\textit{\ }left to the reader to check) that the aggregation rules
for $(N,\mathcal{B}_{A})$ resulting from idempotent ones by just \textit{%
removing cycles }from their outputs through a \textit{minimal} number of
pair-deletions are also $B_{\mu ^{\prime \prime }}$-monotonic (though, of
course, not idempotent but rather \textit{weakly idempotent }in the
following sense: for any profile $R_{N}\in $ $(\mathcal{B}_{A})^{N}$ such
that $R_{i}=R_{j}=R$ for all $i,j\in N$, $f(R_{N})\subseteq R$). Thus, here
is a new (sub)class of interesting strategy-proof aggregation rules for $(N,%
\mathcal{B}_{A})$ whose output for any profile of total preorders is indeed
a total preorder (let us call them \textit{minimal} \textit{monotonic
retracts }just for ease of reference).

Furthermore, for an odd-sized $N$ the majority rule for $(N,\mathcal{B}_{A})$
turns out to coincide with the co-majority rule. This is made precise by the
following:\smallskip

\textbf{Proposition 5.} Let $A$ be a nonempty finite set of alternative
social states, $\mathcal{B}_{A}\in \left \{ \mathcal{B}_{A}^{r},\mathcal{B}%
_{A}^{ir}\right \} $, $\mathcal{X}^{\prime }:=(\mathcal{B}_{A}^{r},\cup
,\cap )$ the (bounded) distributive lattice induced on $\mathcal{B}_{A}$ by $%
\cup $ and $\cap $, $\mu ^{\prime }$ its median ternary operation and $%
B_{\mu ^{\prime }}$ the corresponding betweenness relation, $N$ a finite set
such that $|N|$ is an \textit{odd number, }and $f:(\mathcal{B}%
_{A})^{N}\rightarrow \mathcal{B}_{A}$ an aggregation rule for $(N,\mathcal{B}%
_{A})$. Then, the following statements are equivalent:

$(i)$\ $f$\ satisfies Anonymity and Bi-Idempotence, and is strategy-proof on 
$D^{N}$ for every rich domain $D\subseteq U_{B_{\mu ^{\prime \prime }}}$ of
locally unimodal preorders w.r.t . $B_{\mu ^{\prime }}$ on $\mathcal{B}_{A}$;

$(ii)$ $f=f_{r}^{maj}=f_{r}^{\partial maj}$;

$(iii)$ $f=f_{\preceq }^{CK^{r}}=f_{\preceq ^{\prime }}^{CK^{r}}$i.e. the
generalized Condorcet-Kemeny aggregation rule for $(N,\mathcal{B}_{A})$ for
any pair of linear orders $\preceq ,\preceq ^{\prime }$ on $\mathcal{B}%
_{A}^{r}$.

\begin{proof}
Immediate from Corollary 1 and Example 3.
\end{proof}

Thus, when $N$ has an odd size, generalized Condorcet-Kemeny aggregation
rules for $(N,\mathcal{R}_{A}^{T})$, $(N,\mathcal{T}_{A})$, $(N,\mathcal{B}%
_{A}^{r})$ and $(N,\mathcal{B}_{A}^{ir})$ are amenable to the same sort of
simple characterization via Anonymity, Bi-Idempotence and Strategy-Proofness
on certain rich single-peaked domains. Moreover, in both cases \textit{%
strict Condorcet-Kemeny rules }are also available.

\section{Concluding remarks}

The results of the present work imply that, at least for an \textit{odd-sized%
} population of agents, even anonymous and weakly neutral social welfare
functions on the \textit{full domain of total preference preorders} on a
finite set \textit{do exist}, and are indeed \textit{strategy-proof} on
suitably defined single-peaked domains of `preferences on preferences' (i.e.
arbitrary rich locally unimodal domains).

Arguably, such social welfare functions may also be regarded as a \textit{%
positive} solution to a suitably reformulated version of the classic
Arrowian preference aggregation problem. Namely, the focus is restricted to 
\textit{strategic} as opposed to \textit{structural }manipulation, and the
Arrowian Independence condition IIA is accordingly replaced with a most
`natural' and milder independence requirement tightly related to the
intrinsic order-theoretic structure of $\mathcal{R}_{A}^{T}$.\footnote{%
Indeed, consider $\widehat{f}^{\partial maj}$ for $(N,A)$ with $|N|=|A|=3$,
and profiles $R_{N},R_{N}^{\prime }$ of linear orders with (under the usual
permutation-based notation for linear orders, and square-bracket notation to
denote indifference):
\par
$R_{1}=R_{1}^{\prime }=xyz$; \  \  \ $R_{2}=R_{2}^{\prime }=yzx$; \  \  \ $%
R_{3}=zxy$, $R_{3}^{\prime }=xzy$.
\par
Note that $R_{3}$ and $R_{3}^{\prime }$ are \textit{adjacent}. Nevertheless,
as it is easily checked, $\widehat{f}^{\partial maj}(R_{N})=[xyz]$, while $%
\widehat{f}^{\partial maj}(R_{N}^{\prime })=x[yz].$%
\par
It is then immediately seen that the co-majority rule (which clearly
satisfies $M_{\mathcal{X}}$-Independence with respect to $\mathcal{X}=(%
\mathcal{R}_{A}^{T},\overline{\cup })$) does \textit{not} satisfy IIA with
respect to $\mathcal{X}$. In\ fact, $R_{i}|\left \{ x,y\right \}
=R_{i}^{\prime }|\left \{ x,y\right \} $ for every $i\in N$. Yet, $yf(R_{N})x$
while \textit{not }$yf(R_{N}^{\prime })x$.} In other words, we have here a 
\textit{first} explicit escape route from Arrow's `impossibility' theorem on
preference aggregation, which relies on retention of the `transitivity plus
totality' format requirement\footnote{%
Notice that the version of Proposition 3 that applies to generalized
tournaments and previously mentioned in this work also implies that dropping
transitivity and retaining just totality for both individual and social
preferences is still another way out of Arrowian impossibility results. It
goes without saying that such an escape route would leave ample scope for
manipulation activities through agenda-structure control.} for preference
relations as combined with a considerable \textit{weakening }of IIA that
relies on the (semi)latticial structure of the set of total preorders. Such
a weakening is totally unrelated to other sorts of weakenings of IIA
previously proposed in the literature including several versions of \textit{%
Positionalist Independence,} as\textit{\ }introduced and discussed by
Hansson (1973) with no reference whatsoever to nonmanipulability issues. One
of the strongest of them, labelled as Strong Positionalist Independence
(SPI) by Hansson himself, requires invariance of aggregate preference
between any two alternatives $x,y$ for any pair of preference profiles whose
restrictions to $\left \{ x,y\right \} $ are identical \textit{whenever for
every agent/voter the supports of the respective closed preference intervals
having }$x$ \textit{and }$y$\textit{\ as their extrema are also identical. }%
SPI has been recently rediscovered -and relabeled as Modified IIA- by Maskin
(2020). Maskin motivates it in terms of resistance to certain sorts of `vote
splitting' effects, hence broadly speaking with reference to manipulation
issues, including strategic manipulation. Notice, however, that what is at
stake in that proposal is strategy-proofness of the `\textit{maximizing'
social choice function }induced by a certain social welfare function (as
opposed to strategy-proofness of the social welfare function itself).

In a similar vein, another weakening of IIA that is even stronger than SPI
has been proposed by Saari under the label `Intensity form of IIA' (IIIA).
IIIA requires invariance of aggregate preference between any two
alternatives $x,y$ for any pair of preference profiles such that\textit{\
for every agent/voter the rank (or score) difference between }$x$ \textit{%
and }$y$\textit{\ is left unchanged }from one profile to the other\textit{\ }%
(see Saari (1995)). Arguably, Saari's IIIA can also be regarded as a
formalization of the criticism of IIA originally advanced by Dahl (1956)
with his advocacy of aggregation rules based on intensity of individual
preferences. Notice that IIIA is indeed satisfied by some positional
aggregation rules such as the Borda Count. Moreover, both IIIA and M$_{%
\mathcal{X}}$-Independence are also satisfied by majority judgment as
discussed in Vannucci (2019), which provides an alternative approach to
include intensity of preferences while preserving strategy-proofness. A
further weakening of IIA in a quite different vein is due to Huang (2014),
under the label \textit{Weak Arrow's Independence }(WIIA). In plain words, a
social welfare function $f$ satisfies WIIA if, for any pair $R_{N}$,$%
R_{N^{\prime }\text{ }}$ of profiles of total preorders and any pair $x$,$y$
of alternatives such that the preferences between $x$ and $y$ of every agent 
$i$ in $N$ are the same in $R_{N}$ and $R_{N^{\prime }}$, the following
condition holds: if $x$ is \textit{strictly} \textit{preferred} to $y$
according to social preference $f(R_{N})$ then $x$ is \textit{preferred }%
(i.e. either strictly preferred or indifferent) to $y$ according to social
preference $f(R_{N}^{\prime })$. Notice the main difference between WIIA and
virtually all of the other weakenings of IIA considered in the present work:
while the other weakenings \textit{strenghten the hypothetical clause} of
IIA and leave its consequent unaltered, WIIA keeps the hypothetical clause
of IIA unaltered and \textit{weakens its consequent}.

As mentioned above, even at a first glance one conspicuous difference
between $M_{\mathcal{X}}$-Independence and SPI (or IIIA and WIIA) stands out
immediately: the former relies heavily on the structure of the outcome set,
while SPI, IIIA and WIIA\ only impinge upon the relevant preference
profiles, completely disregarding any specific feature/structure of the
relevant outcome set (namely the set of all total preorders of the set $A$
of basic alternatives). The $M_{\mathcal{X}}$-Independence condition makes
most sense if (to the contrary of IIA) it is divorced from agenda
manipulation issues, and relies in fact on a \textit{fixed agenda} setting.
From a mechanism-design perspective, it amounts to a sort of \textit{%
divide-and-conquer }approach to collective choice problems: the
agenda-formation process and the underlying protocols (if any) are to be
taken for granted, and thereby ignored.\footnote{%
That is not meant to imply that disentangling structural and strategic
manipulation is always easy or indeed possible in actual practice. For
instance, if alternative outcomes are candidates for an appointment or in a
political election then strategic candidacy is virtually always possible.
But strategic candidacy may be regarded precisely as a structural
manipulation of the aggregation rule which is channelled through a forced
change of \textit{available} strategies. Anyway, it is worth noticing that
under social choice functions which admit Condorcet winners at any profile
of their domain \textit{and select} \textit{them,} no agent can benefit by
giving up her own candidacy (see Dutta, Jackson, Le Breton (2001),
Proposition 1).}

\section{Appendix A}

\begin{theorem}
Let $\mathcal{X}=(X,\leqslant )$ be a finite \textit{median
join-semilattice, }$B_{\mu }\mathcal{\ }$its median-induced betweenness, and 
$f:X^{N}\rightarrow X$ an aggregation rule for $(N,X)$. Then, the following
statements are equivalent:

$(i)$ $f$ is strategy-proof on $D^{N}$ for any rich domain $D\subseteq
U_{B_{\mu }}$ of locally unimodal preorders w.r.t. $B_{\mu }$ on $X$;

$(ii)$ $f$ is $B_{\mu }$-monotonic\textbf{;}

$(iii)$ $f$ is monotonically $M_{\mathcal{X}}$-independent.
\end{theorem}

\begin{proof}
$(i)\Longrightarrow (ii)$ By contraposition. Let us assume that $%
f:X^{N}\rightarrow X$ is \textit{not }$B_{\mu }$-monotonic. Thus, there
exist $i\in N$, $x_{i}^{\prime }\in X$ and $x_{N}=(x_{i})_{i\in N}\in X^{N}$
such that $f(x_{N})\notin \lbrack x_{i},f(x_{i}^{\prime },x_{N\backslash
\left \{ i\right \} })]$. Then, consider a preorder $\succcurlyeq ^{\ast }$
on $X$ defined as follows: $x_{i}=top(\succcurlyeq ^{\ast })$ and for all $%
y,z\in X\backslash \left \{ x_{i}\right \} $ , $y\succcurlyeq ^{\ast }z$ if
and only if $(a)$ $\left \{ y,z\right \} \subseteq \lbrack
x_{i},f(x_{i}^{\prime },x_{N\backslash \left \{ i\right \} })]\backslash
\left \{ x_{i}\right \} $ or $(b)$ $y\in \lbrack x_{i},f(x_{i}^{\prime
},x_{N\backslash \left \{ i\right \} })]\backslash \left \{ x_{i}\right \} $
and $z\notin \lbrack x_{i},f(x_{i}^{\prime },x_{N\backslash \left \{
i\right
\} })] $ or $(c)$ $y\nparallel z$, $y\notin \lbrack
x_{i},f(x_{i}^{\prime },x_{N\backslash \left \{ i\right \} })]$ and $z\notin
\lbrack x_{i},f(x_{i}^{\prime },x_{N\backslash \left \{ i\right \} })]$.
Clearly, by construction, $\succcurlyeq ^{\ast }$ consists of three
indifference classes with $\left \{ x_{i}\right \} $, $[x_{i},f(x_{i}^{%
\prime },x_{N\backslash \left \{ i\right \} })]\backslash \left \{
x_{i}\right \} $ and $X\backslash \lbrack x_{i},f(x_{i}^{\prime
},x_{N\backslash \left \{ i\right \} })]$ as top, medium and bottom
indifference classes, respectively. Now, observe that $\succcurlyeq ^{\ast
}\in U_{B_{\mu }}$(which is by definition a rich locally unimodal domain
w.r.t. $B_{\mu }$). To check that such a statement holds true, take any $%
y,z,v\in X$ such that $y\neq z$ and $v\in \lbrack y,z]$, i.e. $\mu (y,v,z)=v$
(if $y=z$, then $v=y=z$ and there is in fact nothing to prove). If $\left \{
y,z\right \} \subseteq \lbrack x_{i},f(x_{i}^{\prime },x_{N\backslash
\left
\{ i\right \} })]$, then by definition $\mu (x_{i},f(x_{i}^{\prime
},x_{N\backslash \left \{ i\right \} }),y)=y$ and $\mu
(x_{i},f(x_{i}^{\prime },x_{N\backslash \left \{ i\right \} }),z)=z$.\ Thus,
by property $(\mu _{2})$ of $\mu $,%
\begin{equation*}
\mu (\mu (x_{i},f(x_{i}^{\prime },x_{N\backslash \left \{ i\right \}
}),y),\mu (x_{i},f(x_{i}^{\prime },x_{N\backslash \left \{ i\right \}
}),z),v)=
\end{equation*}%
\begin{equation*}
=\mu (\mu (y,z,v),x_{i},f(x_{i}^{\prime },x_{N\backslash \left \{ i\right \}
})\text{,}
\end{equation*}
whence%
\begin{equation*}
\mu (\mu (x_{i},f(x_{i}^{\prime },x_{N\backslash \left \{ i\right \}
}),y),\mu (x_{i},f(x_{i}^{\prime },x_{N\backslash \left \{ i\right \}
}),z),v)=\mu (y,z,v)=v
\end{equation*}
implies%
\begin{equation*}
\mu (\mu (y,z,v),x_{i},f(x_{i}^{\prime },x_{N\backslash \left \{ i\right \}
})=\mu (y,z,v)=v\text{,}
\end{equation*}
i.e. $v\in \lbrack x_{i},f(x_{i}^{\prime },x_{N\backslash \left \{
i\right
\} })]$.

Clearly, $\left \{ y,z\right \} \neq \left \{ x_{i}\right \} $ since $y\neq
z $. Now, assume without loss of generality that $y\neq x_{i}$ : thus $%
v\succcurlyeq ^{\ast }y$ by definition of $\succcurlyeq ^{\ast }$. If on the
contrary $\left \{ y,z\right \} \cap (X\backslash \lbrack
x_{i},f(x_{i}^{\prime },x_{N\backslash \left \{ i\right \} })])\neq
\varnothing $, then clearly by definition of $\succcurlyeq ^{\ast }$ there
exists $w\in \left \{ y,z\right \} $ such that $v\succcurlyeq ^{\ast }w$.
Thus, $\succcurlyeq ^{\ast }\in $ $U_{B_{\mu }}$ as claimed. Also, by
assumption $f(x_{N})\in X\backslash \lbrack x_{i},f(x_{i}^{\prime
},x_{N\backslash \left \{ i\right \} })]$ whence by construction $%
f(x_{i}^{\prime },x_{N\backslash \left \{ i\right \} })\succ ^{\ast
}f(x_{N}) $. But then, $f$ is \textit{not }strategy-proof on $U_{B_{\mu
}}^{N}$.\smallskip

$(ii)\Longrightarrow (i)$ Conversely, let $f$ be $B_{\mu }$-monotonic. Now,
consider any $\mathbf{\succcurlyeq }=(\succcurlyeq _{j})_{j\in N}\in
U_{B_{\mu }}^{N}$ and any $i\in N$. By definition of $B_{\mu }$%
-monotonicity, $f(top(\succcurlyeq _{i}),x_{N\backslash \left \{ i\right \}
})\in $ $[top(\succcurlyeq _{i}),f(x_{i},x_{N\backslash \left \{ i\right \}
})] $ for all $x_{N\backslash \left \{ i\right \} }\in X^{N\backslash
\left
\{ i\right \} }$ and $x_{i}\in X$. But then, since clearly $%
top(\succcurlyeq _{i})\succcurlyeq _{i}f(top(\succcurlyeq
_{i}),x_{N\backslash \left \{ i\right \} })$, either $f(top(\succcurlyeq
_{i}),x_{N\backslash \left \{ i\right \} })=top(\succcurlyeq _{i})$ or 
\textit{not }$f(x_{i},x_{N\backslash \left \{ i\right \} })\succ _{i}$%
\textit{\ }$f(top(\succcurlyeq _{i}),x_{N\backslash \left \{ i\right \} })$
by local unimodality of $\succcurlyeq _{i}$w.r.t. $B_{\mu }$. Hence, \textit{%
not }$f(x_{i},x_{N\backslash \left \{ i\right \} })\succ
_{i}f(top(\succcurlyeq _{i}),x_{N\backslash \left \{ i\right \} })$ in any
case. It follows that $f$ is indeed strategy-proof on $U_{B_{\mu }}^{N}$%
.\smallskip

$(ii)\Rightarrow (iii)$ Suppose that $f$ is $B_{\mu }$-monotonic\textbf{. }%
Hence, for all $i\in N$, $y_{i}\in X$, and $(x_{j})_{j\in N}\in X^{N}$, $\
f((x_{j})_{j\in N})\in I^{\mu }(x_{i},\ f(y_{i},(x_{j})_{j\in N\backslash
\left \{ i\right \} })$, i.e. $\ f((x_{j})_{j\in N})=\mu
(x_{i},f((x_{j})_{j\in N}),f(y_{i},(x_{j})_{j\in N\backslash \left \{
i\right \} }))$. Therefore, for any meet-irreducible element $m\in M_{%
\mathcal{X}}$,$\ f((x_{j})_{j\in N})\leqslant m$ if and only if%
\begin{equation*}
\mu (x_{i},f((x_{j})_{j\in N}),f(y_{i},(x_{j})_{j\in N\backslash \left \{
i\right \} }))=
\end{equation*}%
\begin{equation*}
=(x_{i}\vee f((x_{j})_{j\in N}))\wedge (f((x_{j})_{j\in N})\vee
f(y_{i},(x_{j})_{j\in N\backslash \left \{ i\right \} }))\wedge (x_{i}\vee
f(y_{i},(x_{j})_{j\in N\backslash \left \{ i\right \} }))\leqslant m.
\end{equation*}

It follows that if $\ f((x_{j})_{j\in N})\leqslant m$ \ then $%
[x_{i}\leqslant m$ or $f(y_{i},(x_{j})_{j\in N\backslash \left \{ i\right \}
})\leqslant m]$. Indeed, suppose that $\ f((x_{j})_{j\in N})\leqslant m$,
yet [ $x_{i}\nleqslant m$ and $f(y_{i},(x_{j})_{j\in N\backslash \left \{
i\right \} })\nleqslant m$]. Then, $(x_{i}\vee f((x_{j})_{j\in
N}))\nleqslant m$, $f((x_{j})_{j\in N})\vee f(y_{i},(x_{j})_{j\in
N\backslash \left \{ i\right \} })\nleqslant m$, and $x_{i}\vee
f(y_{i},(x_{j})_{j\in N\backslash \left \{ i\right \} })\nleqslant m$.
Therefore, since $\mathcal{X}$ is upper distributive, \ $(x_{i}\vee
f((x_{j})_{j\in N}))\wedge (f((x_{j})_{j\in N})\vee f(y_{i},(x_{j})_{j\in
N\backslash \left \{ i\right \} }))\nleqslant m$\ whence, by upper
distributivity again, $(x_{i}\vee f((x_{j})_{j\in N}))\wedge
(f((x_{j})_{j\in N})\vee f(y_{i},(x_{j})_{j\in N\backslash \left \{
i\right
\} }))\wedge (x_{i}\vee f(y_{i},(x_{j})_{j\in N\backslash \left \{
i\right
\} }))\nleqslant m$,\ i.e. $\mu (x_{i},f((x_{j})_{j\in
N}),f(y_{i},(x_{j})_{j\in N\backslash \left \{ i\right \}
}))=f((x_{j})_{j\in N})\nleqslant m$, a contradiction.

Now, suppose that $m\in M_{\mathcal{X}}$, $f(x_{N})\leqslant m$ and $%
N_{m}(x_{N})\subseteq N_{m}(y_{N})$ for some $x_{N}:=(x_{j})_{j\in N}$, $%
y_{N}:=(y_{j})_{j\in N}\in X^{N}$: we need to establish the claim that $%
f(y_{N})\leqslant m$ as well.

By $B_{\mu }$-monotonicity of $f$, $x_{i}\leqslant m$ or $%
f(y_{i},(x_{j})_{j\in N\backslash \left \{ i\right \} })\leqslant m$ for any 
$i\in N$. Thus, if $x_{i}\leqslant m$, then also $y_{i}\leqslant m$, by
assumption. Hence, $f((x_{j})_{j\in N})\leqslant m$ and $B_{\mu }$%
-monotonicity of $f$ entail $f(y_{i},(x_{j})_{j\in N\backslash \left \{
i\right \} })\leqslant m$: indeed, by $B_{\mu }$-monotonicity $%
f(y_{i},(x_{j})_{j\in N\backslash \left \{ i\right \} })\in I^{\mu }(y_{i},\
f((x_{j})_{j\in N}))$, i.e.%
\begin{equation*}
f(y_{i},(x_{j})_{j\in N\backslash \left \{ i\right \} })=\text{ }
\end{equation*}%
\begin{equation*}
=(y_{i}\vee f(y_{i},(x_{j})_{j\in N\backslash \left \{ i\right \} }))\wedge
(f(y_{i},(x_{j})_{j\in N\backslash \left \{ i\right \} })\vee
f((x_{j})_{j\in N}))\wedge (y_{i}\vee f((x_{j})_{j\in N}))\leqslant
\end{equation*}%
\begin{equation*}
\leqslant (y_{i}\vee f((x_{j})_{j\in N}))\leqslant m\text{.}
\end{equation*}%
It follows that $f(y_{i},(x_{j})_{j\in N\backslash \left \{ i\right \}
})\leqslant m$ in any case.

But then, from $B_{\mu }$-monotonicity of $f$ and $f(y_{i},(x_{j})_{j\in
N\backslash \left \{ i\right \} })\leqslant m$, it similarly follows that $%
x_{i+1}\leqslant m$ or $f((y_{i},y_{i+1},(x_{j})_{j\in N\backslash \left \{
i,i+1\right \} }))\leqslant m$. Now, $x_{i+1}\leqslant m$ entails $%
y_{i+1}\leqslant m$ as well, hence $f(y_{i},(x_{h})_{h\in N\backslash
\left
\{ i\right \} })\leqslant m$ and $B_{\mu }$-monotonicity jointly
imply $f((y_{i},y_{i+1},(x_{j})_{j\in N\backslash \left \{ i,i+1\right \}
}))\leqslant m$, by the same argument previously employed. Repeating the
argument, we eventually obtain $f((y_{i})_{i\in N})\leqslant m$, which
implies that $f$ is indeed \textit{monotonically }$M_{\mathcal{X}}$\textit{%
-independent} as required.\smallskip

$(iii)\Longrightarrow (ii)$ Suppose that $f$ is monotonically $M_{\mathcal{X}%
}$-independent but \textit{not }$B_{\mu }$-monotonic. Thus, there exist $%
i\in N$, $(x_{j})_{j\in N}\in X^{N}$, $y_{i}\in X$ such that $%
f((x_{j})_{j\in N})\neq \mu (x_{i},f((x_{j})_{j\in N}),f(y_{i},(x_{j})_{j\in
N\backslash \left \{ i\right \} }))$, i.e. there must exist $m\in M_{%
\mathcal{X}}$ such that $f((x_{j})_{j\in N})\leqslant m$ but $(x_{i}\vee
f((x_{j})_{j\in N}))\wedge (f((x_{j})_{j\in N})\vee f(y_{i},(x_{j})_{j\in
N\backslash \left \{ i\right \} })\wedge (x_{i}\vee f(y_{i},(x_{j})_{j\in
N\backslash \left \{ i\right \} }))\nleqslant m$ or $(x_{i}\vee
f((x_{j})_{j\in N}))\wedge (f((x_{j})_{j\in N})\vee f(y_{i},(x_{j})_{j\in
N\backslash \left \{ i\right \} })\wedge (x_{i}\vee f(y_{i},(x_{j})_{j\in
N\backslash \left \{ i\right \} }))\leqslant m$ but $f((x_{j})_{j\in
N})\nleqslant m$.

Thus, suppose that $f((x_{j})_{j\in N})\leqslant m$ and $(x_{i}\vee
f((x_{j})_{j\in N}))\wedge (f((x_{j})_{j\in N})\vee f(y_{i},(x_{j})_{j\in
N\backslash \left \{ i\right \} })\wedge (x_{i}\vee f(y_{i},(x_{j})_{j\in
N\backslash \left \{ i\right \} }))\nleqslant m$. Then, it must be the case
that $x_{i}\nleqslant m$ and $f(y_{i},(x_{j})_{j\in N\backslash \left \{
i\right \} })\nleqslant m$ whence by construction $N_{m}((x_{j})_{j\in
N})\subseteq N_{m}((y_{i},(x_{j})_{j\in N\backslash \left \{ i\right \} }))$
and therefore $f(y_{i},(x_{j})_{j\in N\backslash \left \{ i\right \}
})\leqslant m$ by monotonic $M_{\mathcal{X}}$-independence, a contradiction.
Next, suppose that $(x_{i}\vee f((x_{j})_{j\in N}))\wedge (f((x_{j})_{j\in
N})\vee f(y_{i},(x_{j})_{j\in N\backslash \left \{ i\right \} })\wedge
(x_{i}\vee f(y_{i},(x_{j})_{j\in N\backslash \left \{ i\right \}
}))\leqslant m $ and $f((x_{j})_{j\in N})\nleqslant m$.

Since, by upper distributivity of $\mathcal{X}$ , it must be the case that
either $(x_{i}\vee f((x_{j})_{j\in N}))\leqslant m$ or $(f((x_{j})_{j\in
N})\vee f(y_{i},(x_{j})_{j\in N\backslash \left \{ i\right \} })\leqslant m$
or else $(x_{i}\vee f(y_{i},(x_{j})_{j\in N\backslash \left \{ i\right \}
}))\leqslant m$, it follows that $(x_{i}\vee f(y_{i},(x_{j})_{j\in
N\backslash \left \{ i\right \} }))\leqslant m$ hence in particular both $%
x_{i}\leqslant m$ and $f(y_{i},(x_{j})_{j\in N\backslash \left \{ i\right \}
})\leqslant m$. Thus,%
\begin{equation*}
N_{j}((y_{i},(x_{j})_{j\in N\backslash \left \{ i\right \} }))\subseteq
N_{j}((x_{j})_{j\in N})\text{ and }f(y_{i},(x_{j})_{j\in N\backslash \left \{
i\right \} })\leqslant m
\end{equation*}%
whence, by monotonic $M_{\mathcal{X}}$-independence, $f((x_{j})_{j\in
N})\leqslant m$, a contradiction again, and the thesis is established.
\end{proof}

\section{Appendix B}

We present here further properties of a median semi-lattice.

A \textit{chain }of poset $\mathcal{X}=(X,\leqslant )$ is a set $Y\subseteq
X $ such that for any \textit{distinct} $u,v\in Y$ either $u\leqslant v$ or $%
v\leqslant u$ holds, and its \textit{length }$l(Y)$ is $|Y|-1$ (where $|Y|$
denote its size). A chain $Y$ of $(X,\leqslant )$ having $x$ as its $%
\leqslant $-minimum and $y$ as its $\leqslant $-maximum\ is \textit{maximal }%
if there is no $z\in X\setminus Y$ such that $x\leqslant z\leqslant y$. For
any $x,y\in X$ such that $x<y$ (i.e. $x\leqslant y$ and \textit{not} $%
y\leqslant x$) the \textit{length }of the order-interval $[x,y]:=\left \{
z\in X:x\leqslant z\leqslant y\right \} $, written $l([x,y])$, is the length
of a (maximal) chain of \textit{maximum }length having $x$ as its $\leqslant 
$-minimum and $y$ as its $\leqslant $-maximum. In particular, $x\in X$ is
said to be \textit{covered }by $y\in X$, written $x\ll y$, iff $x<y$ and $%
[x,y]=\left \{ x,y\right \} $, namely $l(\left \{ x,y\right \} )=1$. The 
\textit{covering graph }$C(\mathcal{X})=(X,E^{\ll })$ of $\mathcal{X}$ is
the undirected graph having $X$ as vertex-set and $E^{\ll }:=\left \{
\left
\{ x,y\right \} \subseteq X:x\ll y\text{ or }y\ll x\right \} $ as
edge-set. A \textit{path }$\pi _{xy\text{ }}$of $C(\mathcal{X})$ connecting
two vertices $x$ and $y$ is a maximal chain $\left \{
z_{0},....,z_{k}\right
\} $ of $\mathcal{X}$ such that $\left \{
z_{0},z_{k}\right \} =\left \{ x,y\right \} $ and $z_{i}\ll z_{i+1}$, for
any $i=1,...,k-1$, and is of \textit{length }$l(\pi _{xy})=k$\textit{. }The
set of all paths of $C(\mathcal{X})$ connecting $x$ and $y$ is denoted by $%
\Pi _{xy}$. A \textit{geodesic }from $x $ to $y$ on $C(\mathcal{X})$ is a
path of\textit{\ minimum length} (or equivalently a \textit{shortest path})
connecting $x$ and $y$. It can be easily proved (and left to the reader to
check) that the \textit{shortest length} function $\delta _{C(\mathcal{X}%
)}:X\times X\rightarrow \mathbb{Z}_{+}$ such that, for any $x,y\in X$, $%
\delta _{C(\mathcal{X})}(x,y):=l(\pi _{xy})$ (where $\pi _{xy}$ is a path of
minimum length in $\Pi _{xy}$) is indeed a \textit{metric }namely for any $%
x,y,z\in X$: (i) $\delta _{C(\mathcal{X})}(x,y)=0$ iff $x=y$, (ii) $\delta
_{C(\mathcal{X})}(x,y)=\delta _{C(\mathcal{X})}(y,x)$, (iii) $\delta _{C(%
\mathcal{X})}(x,y)\leq \delta _{C(\mathcal{X})}(x,z)+\delta _{C(\mathcal{X}%
)}(z,y)$.

The following properties of finite upper-distributive semi-lattices are also
to be recalled:

\begin{claim}
$(i)$ a finite upper distributive join-semilattice $\mathcal{X}=(X,\leqslant
)$ is \textbf{graded }i.e. it admits a \textbf{rank function} namely a
function $r:X\rightarrow \mathbb{Z}_{+}$ such that for any $x,y\in X$ if $%
x\ll y$ then $r(y)=r(x)+1$ (see Barbut, Monjardet (1970), Leclerc (1994));

$(ii)$ the rank function of a finite upper distributive join-semilattice is
a \textbf{valuation}, namely for any $x,y\in X$ \ the following condition
holds: \textit{if }the meet $x\wedge y$ exists then $\ r(x)+r(y)=r(x\vee
y)+r(x\wedge y)$ (Leclerc (1994)).
\end{claim}

In order to appreciate the actual content of the co-majority rule and its
specialization to the case of total preorders , some further properties of
median join-semilattices are to be introduced and discussed here.

\begin{claim}
(Barth\'{e}lemy (1978), Monjardet (1981), Leclerc (1994)). Let $\mathcal{X}$ 
$=(X,\leqslant )$ be a finite upper distributive join-semilattice, $1$ its
top element, $C(\mathcal{X})$ its covering graph, and $r$ its normalized
rank function defined as follows: for any $x\in X$, $r(x):=r(1)-l([x,1])$.
Then,

$(i)$ the function $d_{r}:X\times X\rightarrow \mathbb{Z}_{+}$ such that for
any $x,y\in X$ $d_{r}(x,y):=2r(x\vee y)-r(x)-r(y)$ is a metric on $X$;

$(ii)$ $d_{r}=\delta _{C(\mathcal{X})}$;

$(iii)$ $\delta _{C(\mathcal{X})}(x,y)=\delta _{C(\mathcal{X})}(x,z)+\delta
_{C(\mathcal{X})}(z,y)$ for any $x,y,z\in X$ such that $z\in \pi _{xy}$ for
some geodesics $\pi _{xy}$ from $x$ to $y$ on $C(\mathcal{X})$.
\end{claim}

\begin{proof}
See Barth\'{e}lemy (1978) (Proposition 1), Monjardet (1981) (Theorem 8), and
Leclerc (1994) (Theorem 3.1).
\end{proof}

As a consequence, for any nonnegative integer $n\in \mathbb{N}$ a nonempty 
\textit{metric median set }$\mathbf{m}(x_{1},...,x_{n})$ can be defined on
any finite family -or \textit{profile- }of $n$ elements $x_{i}$, $i=1,...,n$
of a finite \textit{median} join-semilattice $\mathcal{X}$ $=(X,\leqslant )$
with rank function $r$ and covering graph $C(\mathcal{X})$ as follows:%
\begin{equation*}
\mathbf{m}(x_{1},...,x_{n}):=\left \{ z\in X:z\in \arg \min_{x\in
X}\sum_{i=1}^{n}d(x,x_{i})\right \} \text{,}
\end{equation*}
where $d=d_{r}=\delta _{C(\mathcal{X})}$ as defined above.

Thus, a \textit{metric median function }$\mathbf{m}:\dbigcup \limits_{n\in 
\mathbb{N}}X^{n}$ $\rightarrow \mathcal{P}(X)$ (where $\mathcal{P}(X)$
denotes the power set of $X$), and all of its restrictions $\mathbf{m}%
_{(n)}:X^{n}$ $\rightarrow \mathcal{P}(X)$ to a fixed $n\in \mathbb{N}$, are
well-defined and \textit{nonempty-valued}. In particular, if $n$ is \textit{%
odd }then it is well-known and easily proved that $\mathbf{m}_{(n)}$ is 
\textit{single-valued, }hence it can also be regarded as an $n$-ary
algebraic \textit{operation} on $X$, written $\widehat{\mathbf{m}}%
_{(n)}:X^{n}$ $\rightarrow X$ (see e.g. Bandelt, Barth\'{e}lemy (1984),
Monjardet, Raderanirina (2004), Hudry, Leclerc, Monjardet, Barth\'{e}lemy
(2009)).

The following well-known key result clarifies the tight connection (indeed,
the equivalence) between the co-majority aggregation rule and the foregoing
restrictions of the metric median function on a finite median
join-semilattice $\mathcal{X}$ $=(X,\leqslant )$.

\begin{claim}
(Bandelt, Barth\'{e}lemy (1984)) Let $\mathcal{X}$ $=(X,\leqslant )$ be a
finite median join-semilattice and $\mathbf{m}$ its metric median function.
Then, for any $n\in \mathbb{N}$, and any $(x_{1},...,x_{n})\in X^{n}$,%
\begin{equation*}
f^{\partial maj}(x_{N}):=\dbigwedge \limits_{S\in \mathcal{W}%
^{maj}}(\dbigvee \limits_{i\in S}x_{i})\in \mathbf{m}_{(n)}(x_{1},...,x_{n})%
\text{.}
\end{equation*}%
+Moreover, if $n$ is odd then%
\begin{equation*}
f^{\partial maj}(x_{N}):=\dbigwedge \limits_{S\in \mathcal{W}%
^{maj}}(\dbigvee \limits_{i\in S}x_{i})=\widehat{\mathbf{m}}%
_{(n)}(x_{1},...,x_{n})\text{.}
\end{equation*}
\end{claim}

\begin{proof}
Immediate, by Proposition 5 and dualization of Corollaries 1 and 2 of
Bandelt, Barth\'{e}lemy (1984).
\end{proof}

Thus, the co-majority rule is essentially the same as a \textit{metric
median rule}, namely a rule that selects a metric median.

Finally, some further facts concerning betweenness relations in finite
median join-semilattices are worth mentioning here in order to appreciate
the naturalness and robustness of the betweenness relation involved in our
main characterization Theorem.

Generally speaking, at least \textit{three} distinct betweenness relations
can be defined in a natural way on any finite\textit{\ graded }%
join-semilattice (see e.g. Sholander (1952, 1954), Avann (1961), Van de Vel
(1993)), namely:

$\left( i\right) $ \textit{median betweenness }$B_{\mu }$ (for all $x,y,z\in
X$, $B_{\mu }(x,y,z)$ iff $\mu (x,y,z)=y$ where is the possibly \textit{%
partial }median operation as defined above);

$\left( ii\right) $ \textit{interval betweenness }$B_{I}$ (for all $x,y,z\in
X$, $B_{I}(x,y,z)$ iff $y\leqslant x\vee z$ and $x\leqslant y$ or $%
z\leqslant y$),

$\left( iii\right) $ \textit{metric betweenness }$B_{d}$ (for all $x,y,z\in
X $, $B_{d}(x,y,z)$ iff $d(x,y)+d(y,z)=d(x,z)$, with $d=d_{r}=\delta _{C(%
\mathcal{X})}$ the interval-length-based metric as defined above).

Now, it turns out that if a finite join-semilattice is \textit{median }then
the relationships among $B_{\mu }$,$B_{I}$ and $B_{d}$ is very tight, as
made precise by the following claim.

\begin{claim}
(Sholander (1952, 1954), Avann (1961), Barbut, Monjardet (1970), Leclerc
(1994)). Let $\mathcal{X}$ $=(X,\leqslant )$ be a finite median
join-semilattice. Then, $B_{I}\subseteq B_{\mu }=B_{d}$. Moreover, if $%
\mathcal{X}$ $=(X,\leqslant )$ is in particular a distributive lattice, then

$(i)\ B_{I}(x,y,z)$ holds if and only if $\ x\wedge z\leqslant y\leqslant
x\vee z$;

$(ii)$ $B_{I}=B_{\mu }=B_{d}$;

$(iii)$ $d_{r}(x,y)=r(x\vee y)-r(x\wedge y)=|(M(x)\setminus (M(y))\cup
(M(y)\setminus M(x))|$ (where $M(z):=\left \{ m\in M_{\mathcal{X}%
}:z\leqslant m\right \} $).
\end{claim}

It is thus confirmed, in particular, that local unimodality (the notion of
single-peakedness introduced above and used in Theorem 1) rests on a very
natural and robust notion of betweenness on the underlying poset $\mathcal{X}
$ $=(X,\leqslant )$, which is in turn tightly anchored to an `intrinsic' 
\textit{metric} of $\mathcal{X}$ itself. It should also be emphasized here
that if $\mathcal{X}$ is in particular a median join-semilattice of binary
relations on a finite set (as discussed in the next section), then point
(iii) of the previous Claim establishes that $d_{r}$ is precisely the
so-called \textit{Kemeny distance }for binary relations as defined below
(see Kemeny (1959)).\smallskip 

\textit{Kemeny distance on binary relations. }Let $A$ be a finite set and $%
(B_{A},\subseteq )$ the poset of all binary relations on $A$.\footnote{%
Observe that $(B_{A}^{r},\subseteq )$ is indeed a distributive lattice since
it is obviously closed with respect to both intersection $\cap $ and union $%
\cup $.} Then the \textit{Kemeny distance }on $(B_{A},\subseteq )$ is the
function $d_{K}:B_{A}\rightarrow \mathbb{Z}_{+}$ defined as follows: for any 
$R,R^{\prime }\in B_{A}$,%
\begin{equation*}
d_{K}(R,R^{\prime }):=|\left \{ (x,y)\in A\times A:xRy\text{ and \textit{not }%
}xR^{\prime }y\right \} \cup \left \{ (x,y)\in A\times A:xR^{\prime }y\text{
and \textit{not }}xRy\right \} |\text{.\footnote{%
Thus, the Kemeny distance is just a specialization of the set-theoretic 
\textit{symmetric difference metric }to binary relations (see e.g. Barbut,
Monjardet (1970)).}}
\end{equation*}

\section{Appendix C}

\textbf{Related literature\smallskip }

The study of aggregation rules for ordered sets, semilattices, and lattices
was pioneered by Monjardet and his co-workers, whose contributions provide
characterizations of several classes of such rules mostly within a fixed
population setting but also, occasionally, within a variable population
framework (see e.g. Barth\'{e}lemy, Monjardet (1981), Bandelt, Barth\'{e}%
lemy (1984), Monjardet (1990),\ Barth\'{e}lemy, Janowitz (1991), Leclerc
(1994), Monjardet, Raderanirina (2004), Hudry et al. (2009)). In particular,
characterizations of the simple majority and co-majority rules (sometimes
also denoted as \textit{`median' rules}) are established in several
latticial and semilatticial settings both as aggregation rules within a
fixed population framework (see e.g. Monjardet (1990)) and as
multi-aggregation rules within a variable population framework (see e.g.
Barth\'{e}lemy, Janowitz (1991), Monjardet, Raderanirina (2004)).

Concerning the special case of preference aggregation, an early
characterization of (a version of) the Condorcet-Kemeny rule, regarded as a 
\textit{multi-aggregation rule} for \textit{linear orders} in a \textit{%
variable population} setting is due to Young, Levenglick (1978). Indeed,
Young and Levenglick prove that the Condorcet-Kemeny multi-aggregation rule
is in fact the unique function $f:\dbigcup \limits_{n\in \mathbb{N}}(%
\mathcal{L}_{A})^{n}\longrightarrow \mathcal{P}(\mathcal{L}_{A})\setminus
\left \{ \emptyset \right \} $ that satisfies the following three
properties: \textit{neutrality}, a version of the \textit{Condorcet principle%
}, and `\textit{consistency}' across committees/electorates (i.e. for any
pair of profiles $R_{N},R_{M}$ such that $N\cap M=\emptyset $, if $%
f(R_{N})\cap f(R_{M})\neq \emptyset $ then $f((R_{N},R_{M}))=f(R_{N})\cap
f(R_{M})$).\footnote{%
In subsequent work (see e.g. Young (1995)), it is emphasized that at any
profile $R_{N}$ of linear orders on a finite $A$ the linear orders selected
by the Condorcet-Kemeny rules can also be regarded as the \textit{maximum
likelihood }rankings according to the evidence provided by $R_{N}$. It
should be noted that Young's argument is quite general and also applies to
wider classes of preference relations on $A$ including the set of all total
preorders $\mathcal{R}_{A}^{T}$ and the set of all reflexive relations $%
\mathcal{B}_{A}^{r}$.
\par
{}}

In a similar vein, but in a much more general setting and building partly on
Barth\'{e}lemy, Janowitz (1991), McMorris, Mulder, Powers (2000) establishes
a further elegant characterization of the \textit{median function }as a 
\textit{multi-aggregation rule\ }$f:\dbigcup \limits_{n\in \mathbb{N}%
}X^{n}\longrightarrow (\mathcal{P}(X)\setminus \left \{ \varnothing
\right
\} ) $ for a median meet-semilattice $(X,\leqslant )$ in a \textit{%
variable population} framework, using suitably generalized counterparts of a
weaker version of Condorcet principle (labelled as $\frac{1}{2}$\textit{%
-Condorcet} \textit{property}) and \textit{`consistency' }across
populations/electorates as presented above, and a very mild \textit{%
`faithfulness'} condition simply requiring $f((x))=\left \{ x\right \} $ for
each $x$ in $X$.

The present paper obviously owes much to that most remarkable body of
literature. Notice, however, that the contributions mentioned above \textit{%
do not consider at all strategy-proofness properties of aggregation rules }%
(or, for that matter, nonmanipulability properties of \textit{any} sort).

Some previous joint works of Nehring and Puppe (see in particular Nehring,
Puppe (2007),(2010)) have also several significant connections to the
present contribution. To be sure, Nehring, Puppe (2007) is mainly concerned
with \textit{strategy-proof social choice functions} as defined on profiles
of total preorders on finite sets. Conversely, Nehring, Puppe (2010)\ is
focussed on an `abstract' class of Arrowian aggregation problems including
preference aggregation and, more specifically, social welfare functions, but
it does \textit{not }address issues concerning their strategy-proofness
properties. However, social choice functions with the tops-only property%
\textit{\footnote{%
A social choice function for $(N,A)$ is a function $f:\mathcal{D}%
^{N}\rightarrow A$ where $\mathcal{D}\subseteq \mathcal{R}_{A}^{T}$: it
satisfies the \textit{tops-only property }if $f(R_{N})=f(R_{N}^{\prime })$
whenever $t(R_{i})=t(R_{i}^{\prime })$ for each $i\in N$, and $%
|t(R_{i})|=|t(R_{i}^{\prime })|=1$ for all $i\in N$ (with $t(R_{i}):=\left \{
x\in A:xR_{i}y\text{ for all }y\in A\right \} $).}} may be regarded as
aggregation rules endowed with a specific domain of total preorders, and the
class of Arrowian aggregation rules considered in Nehring, Puppe (2010) does
include the case of preference aggregation rules in finite median
semilattices. Specifically, Nehring and Puppe attach to any \textit{finite }%
outcome space a certain finite hypergraph $\mathbb{H}=(X,\mathcal{H})$
denoted as \textit{property space, }where the set $\mathcal{H}\subseteq 
\mathcal{P}(X)\smallsetminus \left \{ \varnothing \right \} $ of (nonempty)
hyperedges or \textit{properties }of outcomes/states in $X$\textit{\ }is 
\textit{complementation-closed }and \textit{separating }(namely $X\setminus
H\in \mathcal{H}$ whenever $H\in \mathcal{H}$, and for every two \textit{%
distinct }$x,y\in X$ there exists $H_{x^{+}y^{-}}\in \mathcal{H}$ such that $%
x\in $ $H_{x^{+}y^{-}}$and $y\notin H_{x^{+}y^{-}}$). Such a property space $%
\mathbb{H}$ models the set of all \textit{binary} properties of outcomes
that are regarded as relevant for the decision problem at hand. Then, a
betweenness relation $B_{\mathbb{H}}\subseteq X^{3}$ is introduced by
stipulating that $B_{\mathbb{H}}(x,y,z)$ holds precisely when $y$ satisfies
all the properties shared by $x$ and $z$.\footnote{%
In particular, a nonempty subset $Y\subseteq X$ is said to be \textit{convex 
} for $\mathbb{H=}(X,\mathcal{H})$ if for every $x,y\in Y$ and $z\in X$, if $%
B_{\mathbb{H}}(x,z,y)$ then $z\in Y$, and \textit{prime }(or a \textit{%
halfspace}) for $\mathbb{H}$ if both $Y$ and $X\setminus Y$ are convex for $%
\mathbb{H}$ and $\left \{ Y,X\setminus Y\right \} \subseteq \mathcal{H}$.}
Moreover, \textit{single-peaked }preference domains on $X$ can be defined
relying on $B_{\mathbb{H}}$. In particular, $B_{\mathbb{H}}$ is said to be 
\textit{median }if for every $x,y,z\in X$ there exists a unique $m_{xyz}\in
X $ such that $B_{\mathbb{H}}(x,m_{xyz},y)$, $B_{\mathbb{H}}(x,m_{xyz},z)$,
and $B_{\mathbb{H}}(y,m_{xyz},z)$ hold.\footnote{%
In that case, $\mathbb{H}$ is said to be a \textit{median property space, }$%
(X,m^{\mathbb{H}})$ (where $m^{\mathbb{H}}:X^{3}\rightarrow X$ is defined by
the rule $m^{\mathbb{H}}(x,y,z)=m_{xyz}$ for every $x,y,z\in X$) is a 
\textit{median algebra}, and for each $u\in X$ the pair $(X,\vee _{u})$
(where $x\vee _{u}y=y$ iff $m^{\mathbb{H}}(x,y,u)$ $=y$ for some $u\in X$)
is a \textit{median join-semilattice} having $u$ as its maximum.} The
following key results are obtained by Nehring and Puppe: $(i)$ the class of
all idempotent social choice functions which are strategy-proof on the
domain of single-peaked preferences thus defined are characterized in terms
of voting by binary issues through a certain combinatorial property\footnote{%
The combinatorial property mentioned in the text is the so-called
`Intersection Property' which requires that for every minimally inconsistent
set of properties, it must be the case that any selection of winning
coalitions for the corresponding binary issues has a non-empty intersection.}
of the families of winning coalitions for the relevant issues and $(ii)$ 
\textit{if the property space is median }then\textit{\ }such combinatorial
property is definitely met, and consequently non-dictatorial neutral and/or
anonymous strategy-proofs aggregation rules including the majority voting
rule are available (Nehring, Puppe (2007), Theorems 3 and 4). Furthermore,
in Nehring, Puppe (2010) the very same theoretical framework is deployed to
analyze preference aggregation and social welfare functions. In particular,
several `classical' properties for social welfare conditions including the
Arrowian Independence of Irrelevant Alternatives (IIA) property can be
reformulated in more general terms which depend on the specification of the
relevant property space\footnote{%
Specifically, given a property space $\mathbb{H}=(\mathcal{R}_{A}^{T},%
\mathcal{H})$, such a generalized $IIA$ for a social welfare function $f:(%
\mathcal{R}_{A}^{T})^{N}\longrightarrow \mathcal{R}_{A}^{T}$ can be defined
as follows: for every $H\in \mathcal{H}$ and $R_{N},R_{N}^{\prime }\in (%
\mathcal{R}_{A}^{T})^{N}$ such that $\left \{ i\in N:R_{i}\in H\right \}
=\left \{ i\in N:R_{i}^{\prime }\in H\right \} $, if $f(R_{N})\in H$ then $%
f(R_{N}^{^{\prime }})\in H$ as well. Of course the original Arrowian version
of such a generalized IIA is obtained by taking $\mathcal{H}:=\left \{
H_{(x,y)}:x,y\in A\right \} $ with $H_{(x,y)}:=\left \{ R\in \mathcal{R}%
_{A}^{T}:xRy\right \} $.}: it follows that several versions of IIA can be
considered. But then, as it turns out, $(iii)$ the \textit{versions of IIA
attached to median property spaces} are consistent with anonymous and
neutral social welfare functions including those induced by majority-based
aggregation rules (Nehring, Puppe (2010), Theorem 4). Interestingly, a
specific example of a median property space for the set of all total
preorders is also provided, namely the one whose issues consist in asking
for each non-empty $Y\subseteq X$ and any total preorder $R$ whether or not $%
Y$ is a \textit{lower contour }of $R$ with respect some outcome $x\in X$.%
\footnote{%
Thus, the property space suggested here is $\mathbb{H}^{\circ }:=\left \{ 
\mathcal{R}_{A}^{T},\mathcal{H}^{\circ }\right \} $, where $\mathcal{H}%
^{\circ }:=\left \{ H_{L}:\varnothing \neq L\subseteq A\right \} $ and $%
H_{L}:=\left \{ 
\begin{array}{c}
R\in \mathcal{R}_{A}^{T}:\text{for some }x\in A \\ 
L=\left \{ y\in A:xRy\right \}%
\end{array}%
\right \} $.}

The overlappings between such results and those presented here are
remarkable, along with some sharp differences which make them mutually
independent. Since finite median semilattices are indeed an example of a
finite median algebra\footnote{%
Specifically, a finite median join-semilattice can be regarded as a generic
instance of a finite median algebra with one of its elements singled out
(that point corresponds to the top element of the semilattice).}, and are
consequently representable as median property spaces\footnote{%
For instance, it is \textit{always} possible represent a (finite) median
algebra as a (finite) property space by taking as properties its \textit{%
prime} sets as defined through its median betweenness (see e.g. Bandelt, Hedl%
\'{\i}kov\'{a} (1983), Theorem 1.5, and note 19 above for a definition of
prime sets). It is important to observe that in general a finite median
algebra or ternary space admits of several representations by distinct
median property spaces (and other non-median as well). By contrast, a
ternary (finite) algebra or space which is not median can only be
represented by (finite) property spaces which are \textit{not} median.}, all
of the Nehring and Puppe's results mentioned above \textit{do apply }to
finite median semilattices as a special case. Notice however that our
results provide a characterization of strategy-proof aggregation rules for
finite median join-semilattices which is both \textit{more explicit} (it
includes a polynomial description of some such rules) and \textit{more
comprehensive }(it is a complete characterization in that it is not limited
to sovereign and idempotent rules). Concerning alternative representations
of the semilattice of total preorders on a finite set, our treatment can
also be translated in terms of a \textit{median} property space, though a 
\textit{different one} from that considered by Nehring and Puppe. In fact,
in our case the set of relevant properties corresponds to the
meet-irreducibles of that semilattice, namely the total preorders having
just \textit{two} indifference classes, or equivalently a binary ordered
classification of outcomes as \textit{good }or \textit{bad}, respectively.
Accordingly, the collection of relevant issues consist in asking, for each
binary good/bad classification of outcomes and any total preorder $R$,
whether the latter is consistent with the given binary classification.%
\footnote{%
Thus, the appropriate version of generalized IIA in our own model is $%
\mathbb{H}^{\ast }:=(\mathcal{R}_{A}^{T},\mathcal{H}^{\ast })$ with $%
\mathcal{H}^{\ast }:=\left \{ 
\begin{array}{c}
H_{A_{1}A_{2}}:A_{1}\neq \varnothing \neq A_{2} \\ 
A_{1}\cap A_{2}=\varnothing \text{, }A_{1}\cup A_{2}=A%
\end{array}%
\right \} $%
\par
$H_{A_{1},A_{2}}:=\left \{ R\in \mathcal{R}_{A}^{T}:R\subseteq
R_{A_{1}A_{2}}\right \} \, \ $and $R_{A_{1}A_{2}}$ is of course the
two-indifference-class total preorders having $A_{1}$ and $A_{2}$ as top and
bottom indifference classes, respectively. Notice that both $\mathbb{H}%
^{\ast }$ and Nehring-Puppe's $\mathbb{H}^{\circ }$as previously defined
(see footnote 22\ above) are \textit{median }property spaces, while the
original Arrowian $\mathbb{H}$ is not.} Summing up, while Nehring and
Puppe's contributions do not address explicitly strategy-proofness issues
for social welfare functions, their approach via property spaces provides an
additional and helpful perspective to appreciate the content and
significance of the results of the present work.

The issue of \textit{strategy-proofness} \textit{for preference aggregation
rules} has been indeed \textit{explicitly }addressed in the previous
literature, but never -to the best of the authors' knowledge- with respect
to the `full' domain of \textit{all} total preorders on a set. Under the
heading `social welfare functions', Bossert and Storcken (1992) study in
fact aggregation rules for \textit{linear} orders on a finite set (hence
what we refer to as \textit{strict social welfare functions}) and their 
\textit{coalitional strategy-proofness }properties with respect to topped
metric total preference preorders (on the set of linear orders) as induced
by a suitably `renormalized' version of the Kemeny distance to be further
discussed below. They prove an impossibility theorem for those coalitionally
strategy-proof\textit{\ }and sovereign \textit{strict }social welfare
functions\textit{\ }that also satisfy a certain condition of independence
from extrema.

Working within a \textit{variable population} framework, Bossert and
Sprumont (2014) offer several possibility results concerning \textit{%
restricted} strategy-proof aggregation rules (mapping profiles of \textit{%
linear orders} on a finite set $A$ into \textit{total} \textit{preorders }on 
$A$)\textit{\ }which are strategy-proof on the domain of topped preferences
(on the set of total preorders) that are single-peaked with respect to the
median betweenness of the distributive lattice of reflexive binary relations
on $A$ (which amounts to an outcome space $\mathcal{B}_{A}^{r}$which is far
more comprehensive than the `small' domain-base $\mathcal{L}_{A}$ or even
the larger codomain $\mathcal{R}_{A}^{T}$ of the aggregation rule\footnote{%
Thus, in a sense, the median-induced betweenness relation under
consideration (and the resulting single-peakedness property) is \textit{not}
the one `naturally' dictated by the codomain $\mathcal{R}_{A}^{T}$ (let
alone the strictly smaller domain-base $\mathcal{L}_{A}$) of the aggregation
rule.}). That paper identifies some (variable-population) strategy-proof 
\textit{restricted }aggregation rules on $\mathcal{R}_{A}^{T}$ including
(strict) Condorcet-Kemeny rules, a class of variable-population counterparts
of our \textit{monotonic retracts of the majority relation }as introduced
above, and a family of rules denoted as \textit{status-quo rules }that are
related to the class of \textit{outcome-biased rules }mentioned above as one
family of examples covered by Proposition 2\textit{.} An (implicit)
characterization of such monotonic majority-retracts is also provided, and
the family of status-quo rules is explicitly characterized (but the strict
Condorcet-Kemeny rules are not). Thus, the present paper provides \textit{%
extensions of the fixed-population counterparts} of such strategy-proof 
\textit{restricted} aggregation rules to strategy-proof \textit{exact}
aggregation rules for \textit{total preorders}, and a \textit{unified }joint
characterization of all of them (see in particular Corollary 1, Propositions
2 and Proposition 3 above), as well as a specific characterization of
generalized Condorcet-Kemeny rules for the case of \textit{odd-dimensional }%
domains\textit{.} Notice, however, that the notion of betweenness underlying
the relevant notion of single-peakedness for `preferences on preferences'
that guarantees the strategy-proofness of such \textit{exact} rules in the
present paper is in fact a most `natural' one, namely the median betweenness
which is characteristic of their domain-base $\mathcal{R}_{A}^{T}$ (but is 
\textit{not} well-defined on its subdomain $\mathcal{L}_{A}$).

The issue of strategy-proof aggregation in \textit{arbitrary }(possibly
infinite) join-semilattices is addressed in Bonifacio, Mass\'{o} (2020)
within a \textit{fixed population }framework. To be sure, that work focuses
in fact on so-called `simple rules', namely anonymous and
unanimity-respecting \textit{social choice functions} with the \textit{%
`tops-only'-property}.\textit{\footnote{%
See footnote 31 above.}} But then, such `simple rules' are essentially
equivalent to \textit{anonymous} and \textit{idempotent} aggregation rules
which are endowed with an \textit{explicitly pre-defined domain of
preference profiles of total preorders}. In particular, the Authors consider
a restriction on total topped preference preorders they denote (join-)`%
\textit{semilattice-single-peakedness}'\footnote{%
The notion of semilattice-single-peakedness (SSP) for total preorders on a
join-semilattice $(X,\leqslant )$ was first introduced in Chatterji, Mass%
\'{o} (2018). A total preorder $R$ on $X$ is SSP in $(X,\leqslant )$ if and
only if : (i) $R$ has a unique maximum element $x^{\ast }$ in $X$; (ii) $yRz$
for each $y,z\in X$ such that $x^{\ast }\leqslant y\leqslant z$; (iii) $%
(x^{\ast }\vee u)Ru$ for each $u\in X$ such that $x\nleqslant u$.} which
results in a maximal domain that is consistent with the existence of
strategy-proof `simple rules'. Then, they proceed to characterize \textit{%
the subclass} of anonymous and idempotent strategy-proof aggregation rules,
establishing that they are precisely the \textit{`supremum' rule }$f^{\vee }$
and a family of `\textit{generalized quota-supremum}' rules.\footnote{%
The \textit{`supremum'} (or \textit{join }$n$-\textit{projection ) }rule $%
f^{\vee }$ for $(N,X)$ is defined as follows: $f^{\vee }(x_{N}):=\vee _{i\in
N}x_{i}.$ A \textit{generalized quota-supremum} rule returns a certain
prefixed alternative $x^{\ast }$ if $x^{\ast }$ reaches a prespecified
quota, and $\vee _{i\in N}x_{i}$ otherwise.} It should also be noticed that
such a comparatively weak notion of semilattice-single-peakedness is
admittedly consistent with the notion of single-peakedness induced by
metric-betweenness according to the shortest-path-metric on the covering
graph of the semilattice.\footnote{%
That is so because (if the join-semilattice $(X,\leqslant )$ is \textit{%
discrete} i.e. it has no bounded infinite chain) for any pair of elements $%
x,y$ $\in X$ which are \textit{not} $\leqslant $-comparable the join $x\vee
y $ must lie on a shortest path from $x$ to $y$ of the covering graph of the
semilattice.}

However, semilattice-single-peakedness is clearly bound to relinquish any
connection not only to a median-induced betweenness if the relevant
semilattice is \textit{not} median, but also to the most \textit{natural}
rank-based metric betweenness if the semilattice also happens to be \textit{%
not} \textit{even graded}.\textit{\footnote{%
Lattices (hence, of course, semilattices) which are not graded are quite
common: in the present context, the lattice of \textit{partial preorders} is
perhaps the most obvious example (see e.g. Barbut, Monjardet (1970)).}}
Therefore, in the latter case there is no natural metric to ground the claim
that a certain type of single-peakedness describes a sort of `preferences on
preferences' that are induced in a `natural' -hence plausibly shared- way by
the \textit{actual} basic preferences of agents.

It is also worth mentioning here that, in any case, strategy-proofness only
concerns \textit{strategic manipulation} of a preference-aggregation
process, namely manipulation of the outcome of a certain game by means of an 
\textit{appropriate choice of strategy} in the available strategy-set(s). In
other terms, a \textit{given game} is implicitly being taken for granted,
including of course the \textit{population} of its players and the set of
its possible alternative outcomes, or its \textit{agenda. }But then, \textit{%
manipulation of the agenda} (or, for that matter, of the relevant population
of players itself) can also be considered: notice, however, that from a
game-theoretic perspective, that is a kind of \textit{structural (as opposed
to strategic) manipulation} since it amounts to a change of the \textit{game
itself. }

Such a broader perspective on manipulation issues in preference-aggregation
is apparent (if mostly implicit) in Sato (2015).\footnote{%
This is also, arguably, Arrow's own perspective on manipulation issues (see
Arrow (1963)), except that he overtly renounces to address \textit{strategic 
}manipulation issues, while acknowledging their substantial import (see e.g.
Arrow (1963), chpt. 1). By contrast, \textit{agenda-manipulation} issues
play a key role in the arguments offered by Arrow to support his own
proposal of the Independence of Irrelevant Alternatives (IIA) condition for
social welfare functions (a more detailed discussion of the relationship of
IIA to agenda manipulation will be provided elsewhere).} Indeed, Sato's
contribution relies on a \textit{fixed population }framework and is mainly
focussed on \textit{strict social welfare functions }as defined on some 
\textit{connected }domain of linear orders over a finite set $A$.\footnote{%
A \textit{connected }domain of linear orders over $A$ is a set $\mathcal{%
D\subseteq L}_{A}$ such that for any $R,R^{\prime }\in \mathcal{D}$ there
exists a finite family $\left \{ R_{1},...,R_{k}\right \} \subseteq \mathcal{%
D} $ such that (i) $R_{1}=R$; (ii) $R_{k}=R^{\prime }$; (iii) for every $%
i=1,...,k-1$, $R_{i}$ and $R_{i+1}$ can be mutually obtained by reversing
the respective ranks of two adjacent (or consecutive) elements of $A$ that
are `adjacent' (i.e. consecutive) according to the other. Thus, a (strict)
social welfare function on a connected domain is a function $f:\mathcal{D}%
^{N}\longrightarrow \mathcal{L}_{A}$ (clearly, it is also \textit{restricted 
}for $(N,\mathcal{L}_{A})$ if $\mathcal{D\neq L}_{A}$). Observe that $%
\mathcal{L}_{A}$ itself is of course a connected domain.} However, it also
considers the family of \textit{social choice functions} which are induced
by any such strict social welfare function on the subsets of $A$ through
maximization -at each preference profile- of the `social' linear order
selected at that profile (as restricted to the relevant subset of $A$). In
that connection, \textit{four} notions of \textit{nonmanipulability }for
strict social welfare functions are considered, with the primary aim to
address issues of \textit{strategic} manipulation.\footnote{%
One of them is akin to the notion of strategy-proofness for aggregation
rules proposed by Bossert, Sprumont (2014) as discussed above, and another
one relies on the `renormalized' Kemeny distance for linear orders. By
contrast, the last two nonmanipulability notions invoke the induced
maximizing choices on $A$, and on its subsets, respectively (and are also
most suitable to address certain agenda-manipulation issues).} Then, relying
on a \textit{renormalized and `contracted' }version $\widehat{d}_{K}$ 
\footnote{%
Namely, the `halved' Kemeny distance for linear orders. That is essentially
the distance between rankings due to Kendall, given by the minimal number of
transpositions of adjacents elements that is necessary to obtain one linear
order starting from another one (see e.g. Kendall (1955)).} of the \textit{%
Kemeny distance } as defined previously, Sato introduces a `continuity-type'
condition for strict social welfare functions called \textit{Bounded Response%
}. A strict social welfare function $f$ satisfies Bounded Response if $%
\widehat{d}_{K}(f(R_{N}),f(R_{N}^{\prime }))\leq 1$ whenever two preference
profiles $R_{N},R_{N}^{\prime }$ are the same except for the preference of a
single agent $i$, and $R_{i}$ and $R_{i}^{\prime }$ are \textit{adjacent}
(i.e. $R_{i}^{\prime }$ is obtained from $R_{i}$ by permuting the $R_{i}$%
-ranks of \textit{a single} pair of alternatives with \textit{consecutive} $%
R_{i}$-ranks).\footnote{%
It is worth recalling here that $1$ is the minimum positive value of \ both $%
d_{K}$ and $\widehat{d}_{K}$.} In a similar vein, a very mild \textit{%
Adjacency-Restricted Monotonicity }condition for strict social welfare
functions is considered. The main result established by Sato (2015) is 
\textit{the equivalence} of the following statements for a strict social
welfare function $f$ on a connected domain of linear orders on $A$: (1) $f$
satisfies Bounded Response and \textit{at least one} of the four
nonmanipulability conditions mentioned above; (2) $f$ satisfies Bounded
Response and \textit{each one }of the four nonmanipulability conditions
mentioned above; (3) $f$ satisfies Adjacency-Restricted Monotonicity and the
Arrowian \textit{Independence of Irrelevant Alternatives (IIA)} condition.%
\textit{\ }

As a corollary of that result (and of arguments from standard proofs of the
Arrowian `impossibility' theorem for strict social welfare functions) a new
characterization of \textit{dictatorial} strict social welfare functions in
terms of Bounded Response, one of the four equivalent nonmanipulability
conditions mentioned above, and Sovereignty (or Ontoness)\footnote{%
A strict social welfare function $f$ is \textit{sovereign }if for any $L\in 
\mathcal{L}_{A}$ there exists $R_{N}\in \mathcal{L}_{A}^{N}$ such that $%
f(R_{N})=L$.} is established. Furthermore, the set $\mathcal{D}^{sp(Q)}$of
linear orders on $A$ which are \textit{single-peaked }with respect to some 
\textit{fixed linear order }$Q$ on $A$ can also be shown to be a \textit{%
connected }domain, and the strict social welfare function $f^{wmaj}$induced
by the method of `(weak) majority decision'\footnote{%
Namely, for any $R_{N}\in (\mathcal{L}_{A})^{N}$, and $x,y\in A$, $%
xf^{wmaj}(R_{N})y$ if and only if $|N_{x}(R_{N})|\geq |N_{y}(R_{N})|$.}
clearly satisfies both Adjacency-Restricted Monotonicity and IIA. Hence, it
immediately follows that $\widetilde{f}^{wmaj}:\mathcal{D}%
^{sp(Q)}\rightarrow \mathcal{L}_{A}$ is a (restricted) strict social welfare
function which satisfies both Bounded Response and all of the four
nonmanipulability conditions mentioned above (hence, in particular, the
strategy-proofness properties implied by the first two conditions from that
list).

Thus, at least when applied to \textit{strict} social welfare functions, the
combination of Bounded Response and standard nonmanipulability conditions
(including, more specifically, strategy-proofness requirements) tends
apparently to reproduce a well-known pattern. Namely, `impossibility'
theorems on the \textit{full} domain of linear orders, and some
`possibility' results on suitably \textit{restricted domains }of linear
orders (to the effect that e.g. several versions of the simple majority rule
provide well-defined and strategy-proof \textit{restricted strict social
welfare functions }on certain single-peaked domains of linear orders%
\footnote{%
The significant body of literature devoted to the elaboration of such two
related themes is extensively reviewed in the fourth chapter of Gaertner
(2001).}).

By contrast, the existence issue for strategy-proof social welfare functions
as aggregation rules on the \textit{full domain of total preorders or even
larger sets of reflexive and possibly nontransitive binary relations} on a
finite set has never been addressed explicitly in previously published work,
as mentioned above.

\end{document}